\RequirePackage{amsmath}
\documentclass[runningheads,11pt]{llncs}
\usepackage{geometry}
\usepackage[utf8]{inputenc}

\usepackage{amsthm}
\usepackage{setspace}
\usepackage{xcolor}
\usepackage{hyperref}
\hypersetup{
     colorlinks=false,
     linkcolor=blue,
     filecolor=blue,
     citecolor = black,  
     urlcolor=cyan,
     urlbordercolor=violet,
     allbordercolors=violet!60,
     }
\usepackage{amsmath}
\usepackage{amssymb}
\usepackage{tikz}
\usepackage{subcaption}
\usepackage{booktabs}
 \usepackage[algoruled,lined,linesnumbered,algonl,procnumbered]{algorithm2e}
\usepackage[size=tiny,color=teal!15]{todonotes}
\usepackage{float}
\usepackage[normalem]{ulem}

\usetikzlibrary{calc}
\usetikzlibrary{arrows,shapes,decorations.pathreplacing,decorations.pathmorphing}

\newcommand{\triangularGrid}[4]{
    \newcommand\rowNum{#1} 
    \newcommand\colNum{#2} 
    \newcommand\triSide{#3} 
    \newcommand\triHeight{#4} 

    \clip (-.5*\triSide,-1) rectangle (\triSide+\triSide*\colNum -.5*\triSide,4*\triHeight*\rowNum) ;

    \foreach \i in {0,...,\rowNum} {
        \ifnum\i=0 
            \foreach \j in {0,...,\colNum} {
                \draw[gray!50] (0:\j*\triSide) -- ++ (-120:\triSide) -- ++ (0:\triSide) --++ (120:\triSide) --++ (0:\triSide) ;
            }
            \foreach \j in {0,...,\colNum} {
                \draw[gray!50] (0:\j*\triSide-\triSide) --++ (0:\triSide) --++ (120:\triSide) --++ (0:\triSide) --++ (-120:\triSide);
            }
        \else
            \foreach \j in {0,...,\colNum} {
                \newcommand\startingPoint{2*\i*\triHeight}
                \draw[gray!50] (90:\startingPoint) ++ (0:\j*\triSide) -- ++ (-120:\triSide) -- ++ (0:\triSide) --++ (120:\triSide) --++ (0:\triSide) ;
            }
            \foreach \j in {0,...,\colNum} {
                \newcommand\startingPoint{2*\i*\triHeight}
                \draw[gray!50] (90:\startingPoint) ++ (0:\j*\triSide-\triSide) --++ (0:\triSide) --++ (120:\triSide) --++ (0:\triSide) --++ (-120:\triSide);
            }
        \fi
    }
}

\begin{document}

\title{Deterministic Leader Election for Stationary Programmable Matter with Common Direction\thanks{This work has been partially supported by ANR project DUCAT (ANR-20-CE48-0006)}} 

\titlerunning{Deterministic Leader Election with Common Direction} 

\author{Jérémie Chalopin \and
Shantanu Das \and
Maria Kokkou }
\authorrunning{J. Chalopin et al.}
%
\institute{
Aix Marseille Univ, CNRS, LIS, Marseille, France\\
\email{\{jeremie.chalopin,shantanu.das,maria.kokkou\}@lis-lab.fr}}

\maketitle

\begin{abstract}
Leader Election is an important primitive for programmable matter, since it is often an intermediate step for the solution of more complex problems. Although the leader election problem itself is well studied even in the specific context of programmable matter systems, research on fault tolerant approaches is more limited. We consider the problem in the previously studied Amoebot model on a triangular grid, when the configuration is connected but contains nodes the particles cannot move to (e.g., obstacles). We assume that particles agree on a common direction (i.e., the horizontal axis) but do not have chirality (i.e., they do not agree on the other two directions of the triangular grid). We begin by showing that an election algorithm with explicit termination is not possible in this case, but we provide an implicitly terminating algorithm that elects a unique leader without requiring any movement. These results are in contrast to those in the more common model with chirality but no agreement on directions, where explicit termination is always possible but the number of elected leaders depends on the symmetry of the initial configuration. Solving the problem under the assumption of one common direction allows for a unique leader to be elected in a stationary and deterministic way, which until now was only possible for simply connected configurations under a sequential scheduler.
\end{abstract}


\section{Introduction}
Programmable Matter (PM) refers to a distributed system consisting of a large number of constant--memory computational entities (called \textit{particles}) evolving in a geometric environment and acting collaboratively in order to accomplish a given task. We consider the task to be the well--known Leader Election (LE) problem, that is, electing a unique leader among the particles. Particles in PM systems do not have unique identifiers or a global sense of direction, making it difficult to break the symmetry and elect a unique leader. In some contexts, particles may be unable to move. For example, this can be so as to maintain a configuration or due to the presence of obstacles (such as foreign objects) in the system. Hence, we study the case of particles electing a leader without any movement, proposing a so-called \textit{stationary} algorithm. Leader Election without movement has been studied in the literature in the context of PM (e.g., \cite{bazzi2019stationary,di2020shape}), however, it is not always possible to deterministically elect a unique leader. As a simple example consider a synchronous system of three particles forming a triangle in a triangular grid. Without additional assumptions, there is no deterministic algorithm that elects a unique leader. Moreover, even if the particles have chirality (i.e., a common notion of clockwise and counterclockwise directions) a unique leader cannot be obtained for every initial configuration by existing algorithms, such as \cite{bazzi2019stationary,dufoulon2021efficient}. Providing a method to elect a unique leader can be useful for problems where a task cannot be performed by multiple leaders but a single leader would make the problem solvable. For example, in \cite{di2020shape} it is proven that $k$ leaders in a $k$--symmetric initial configuration of particles cannot perform shape formation without additional capabilities, unless the target shape is also $k$--symmetric. We show that adding agreement on one direction instead of chirality suffices to elect a unique leader. We give an implicitly terminating algorithm where a unique particle eventually switches its state to leader and no particle ever changes state again, but particles never know whether the algorithm has terminated. We prove that explicit termination is not possible for particles agreeing on one direction without any additional capabilities. However, complimentary to the case of particles with common chirality, a unique leader is eventually elected in our setting even under an asynchronous but fair scheduler.

\subsection{Related Work and Motivation} \label{sec:related-work-motivation} 
The concept of PM was introduced and formalized in \cite{toffoli1991programmable}. Since then, various very distinct models such as \cite{AM-derakhshandeh2014amoebot,hawkes2010programmable,bourgeois2016programmable,woods2013active,fekete2021cadbots,feldmann2022coordinating} have been proposed. We use the Amoebot model (Section \ref{sec:model-preliminaries}), introduced in \cite{AM-derakhshandeh2014amoebot} and updated in \cite{daymude2023canonical}. Leader Election has been studied in both a two dimensional setting, as we describe below, and in three dimensions in \cite{gastineau2022leader,briones2023asynchronous}. All LE algorithms for deterministic 2D settings, including the one in this paper, assume that particles operate in a geometric environment (i.e., triangular grid), have constant memory and are activated by a fair scheduler. However, each of the existing algorithms uses at least two additional assumptions or capabilities. We list the different options for those model choices here and underline the ones used in our work: 
\begin{itemize}
    \item Using a randomized algorithm to break symmetries (e.g., \cite{daymude2017improved}) or a \underline{deterministic} algorithm.
    \item  Electing $k \in \{1,2,3,6\}$ leaders in $k$-symmetric configurations (e.g., \cite{di2020shape}) or a \uline{unique leader} regardless of the configuration.
    \item A simply connected configuration (e.g., \cite{gastineau2019distributed}) or a configuration containing \underline{holes}.
    \item Particles with (e.g., \cite{bazzi2019stationary}) or \underline{without common chirality} which is a common sense of rotational orientation.
    \item Particles that can move from a node to a neighbouring one (e.g., \cite{emek2019deterministic}) or \uline{stationary particles}.
    \item A sequential scheduler (e.g., \cite{dufoulon2021efficient}) which activates one particle at a time and the particle finishes its action or an \underline{asynchronous scheduler} which simultaneously activates any number of particles. It is worth noting that using the concurrency control framework presented in \cite{daymude2023canonical}, algorithms which terminate under a sequential scheduler and satisfy certain conditions can be transformed into equivalent algorithms that work in the concurrent setting using \textit{locks}. However, even if we remove the sequential scheduler assumption in previous stationary LE algorithms given the results of \cite{daymude2023canonical}, current stationary algorithms cannot elect a unique leader without additional capabilities.
\end{itemize}
We summarise the results and assumptions of previous work in Table \ref{tab:previous-LE}. 
\begin{table}[h]
    \begin{center}
    \resizebox{\columnwidth}{!}{%
    \begin{tabular}{c|c|c|c|c|c|c|c}
        \toprule
        Paper & Leaders & Simply Connected Particles & Chirality & Movement & Seq. Scheduler & One Direction & Dimension \\ \toprule
        \cite{gastineau2019distributed} & 1 & \checkmark & \checkmark & \texttt{X} & \checkmark & \texttt{X} & 2D \\ 
        \cite{emek2019deterministic} & 1 & \texttt{X} & \texttt{X} & \checkmark & \checkmark & \texttt{X} & 2D \\
        \cite{dufoulon2021efficient} & 1 & \texttt{X} & \checkmark & \checkmark & \checkmark & \texttt{X} & 2D \\
        \cite{di2020shape} & 3 & \checkmark & \texttt{X} & \texttt{X} & \texttt{X} & \texttt{X} & 2D \\ 
        \cite{dufoulon2021efficient} & 6 & \texttt{X} & \checkmark & \texttt{X} & \checkmark & \texttt{X} & 2D \\
        \cite{bazzi2019stationary} & 6 & \texttt{X} & \checkmark & \texttt{X} & \texttt{X} & \texttt{X} & 2D \\  
        \cite{briones2023asynchronous} & 1 & \checkmark & \texttt{X} & \texttt{X} & \checkmark & \texttt{X} & 2D \& 3D\\
        \midrule
        \textbf{This Work} & 1 & \texttt{X} & \texttt{X} & \texttt{X} & \texttt{X} & \checkmark & 2D \\
        \bottomrule
    \end{tabular}
    }
    \end{center}
    \caption{ Deterministic LE in two-dimensions. In all papers papers particles operate in a triangular grid (see Section \ref{sec:model-preliminaries}). ``Simply Connected Particles'' refers to the assumption that the particle system does not have \textit{Holes}. ``Chirality'' is a common sense of rotational orientation. ``Movement'' is the ability of particles to move from a node to a neighbouring one. A ``Sequential Scheduler'' activates one particle at a time. ``One Direction'' is a common sense of orientation on one direction on the grid.}
    \label{tab:previous-LE} 
\end{table}

We are interested in determining the minimum capabilities that particles need to deterministically elect a \textit{unique} leader in a 2D system. To the best of our knowledge the case of particles agreeing on one direction instead of having chirality or movement capabilities has not been studied before. Therefore, an additional motivation was to determine the differences between assuming chirality and assuming common direction. We show that the difficulties that arise from each of the assumptions are different and the results that can be obtained are complementary to each other. In previous work like \cite{bazzi2019stationary} it was shown that assuming chirality we can get an explicitly terminating algorithm, but up to six leaders are elected depending on the symmetry of the configuration. We show here that agreement on one direction allows for a single leader to be elected, however, an algorithm with explicit termination is no longer possible (Section \ref{sec:termination-impossibility}). Hence, perhaps surprisingly, the assumptions of common chirality and common direction are not directly comparable with respect to which one is more general. 

From a practical perspective, we find scenarios where a large number of particles can autonomously detect the intended common direction more natural, whereas, agreement on chirality is potentially a less intuitive capability to implement. For example, particles can be placed at a slightly tilted plane in order to collectively agree on a common direction, similarly to the setting in \cite{becker2014reconfiguring}. Alternatively, a large enough source of light placed at a sufficiently long distance can provide a common direction for light detecting particles. A setting with some similarities to this case was described in \cite{savoie2018phototactic}.

\subsection{Our Contributions} \label{sec:contributions}
We provide the first result on deterministically electing a unique leader without using movement or chirality, assuming that the particles agree on one direction and that the initial configuration is connected. We show that explicitly terminating LE is not possible for this case (Section \ref{sec:termination-impossibility}) but we give an algorithm that elects a unique leader in $O(n^3)$ rounds, where $n$ is the number of particles in the system (Section \ref{sec:algorithm-description}). Our results are complementary to the case of particles only agreeing on chirality, such as in \cite{bazzi2019stationary}.

Our algorithm (Section \ref{sec:algorithm-description}) uses a message passing procedure to elect intermediate leaders on boundaries. In Section \ref{sec:messages-on-boundaries} we show that if a pair of particles on a boundary can \textit{detect} each other's chirality, they can differentiate between the boundaries in which they participate. We then give a message passing procedure for such particles, so that a message originating at a particle on some boundary $B$ is only received by particles also on $B$. Once intermediate leaders are elected, the system is encoded as a set of trees, each rooted at one of the leaders, and the trees are compared and merged until only one root remains. Encoding parts of a network as trees in order to compare them, is a technique that is used in both mobile agent computing (e.g., \cite{das2006effective}) and in PM systems (e.g., \cite{di2020shape}). However, here we mark the endpoints of the edges where a comparison occurs in the encoding of the tree, adding a dynamic aspect to the system encoding.  

\section{Model and Preliminaries} \label{sec:model-preliminaries}
Let $G_{\Delta}$ be an infinite regular triangular grid. We use two coordinates, $x$ and $y$, to describe the relative position of nodes in the grid. Suppose $v_1 = (x_1, y_1)$ and $v_2$ are adjacent nodes of $G_{\Delta}$. If $v_2$ is on the East (resp. West) of $v_1$, $v_2 = (x_1 + 2, y_1)$ (resp. $v_2 = (x_1 - 2, y_1)$). If $v_2$ is on the North East (resp. South West) of $v_1$, $v_2 = (x_1 + 1, y_1 + 1)$ (resp. $v_2 = (x_1 -1, y_1 - 1)$). Finally, if $v_2$ is on the North West (resp. South East) of $v_1$, $v_2 = (x_1 - 1, y_1 + 1)$ (resp. $v_2 = (x_1 + 1, y_1 - 1)$). The particle system consists of a finite, connected subset of $G_\Delta$, such that each of the nodes in this subset contains exactly one particle. We refer to nodes occupied by particles as \textit{particle nodes} and to nodes not occupied by particles as \textit{non--particle nodes}. Each particle has six ports associated with a \textit{port number} in $\{0, \ldots, 5\}$, corresponding to each edge leading to a neighbouring node. We define ports 0 (East) and 3 (West) to be common for all particles. The remaining ports are labelled in a circular way such that port $i$ and port $i + 1 \mod 6$ lead to neighbouring nodes and the local labelling of the ports is known to the particle. We refer to ports with a port number in $\{0,1,5\}$ (resp. $\{2,3,4\}$) as \textit{right} (resp. \textit{left}) ports. For directions that are not 0 or 3, the port numbers are not consistent, even among neighbours. The particles have constant memory and they cannot move, but they can communicate with particles at distance one by exchanging messages. A message, $m$, sent by an active particle $p$ to a neighbouring particle $p'$ is received when $p'$ is activated and both $p$ and $p'$ know the ports $m$ is sent and received through. We assume that particles are activated by a \textit{fair asynchronous} scheduler so that every particle is activated infinitely often and it is possible but not necessary that all particles are activated simultaneously. 

We call an edge $uv$ a \textit{local boundary} if there exists a non--particle node $o$ that is a common neighbour of $u$ and $v$ and we write $(uv, o)$ to denote the local boundary. A sequence of local boundaries $(u_1u_2, o_1), (u_2u_3, o_2), \ldots, (u_ku_1, o_k)$ is a \textit{boundary} when moving around $u_i$ from $u_{i-1}$ to $u_{i+1}$ only non--particle nodes starting with $o_{i-1}$ and ending with $o_i$ are encountered. Notice that it is possible that $o_{i-1} = o_i$. For any $(uv,o)$ we say that $u$ and $v$ are neighbours on the local boundary. We call particles that are on at least one boundary and are connected to each other through their 0 and 3 ports, a \textit{horizontal path}. An edge between two particles, $p_a$ and $p_b$, on a horizontal path in which $p_a$ and $p_b$ do not have a common neighbouring particle, is called a \textit{dark blue edge} (DBE). An edge between $p_a$ and $p_b$ in a horizontal path where $p_a$ and $p_b$ have one common neighbouring particle is called a \textit{light blue edge} (LBE). All other edges are \textit{grey edges}. A subsystem of particles connected by grey edges and LBEs but not by DBEs, is called a \emph{grey component}. An example of the different kind of edges is shown in Figure \ref{fig:example-all-edges}.
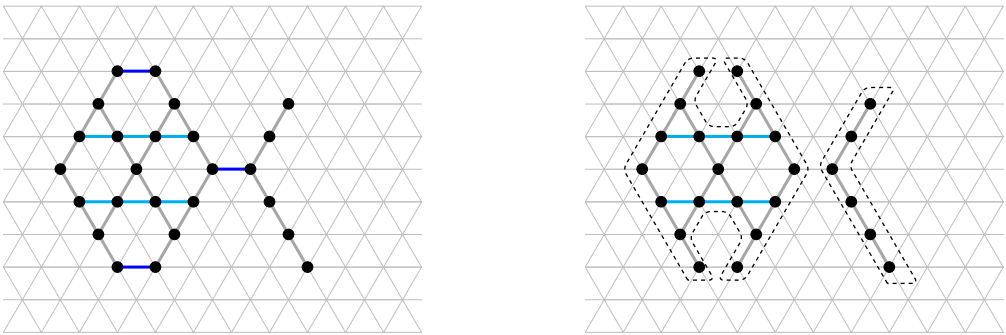
\begin{figure}[h]  
    \begin{minipage}{.49\textwidth}
    \centering
    \vspace{-2.6cm}
    \scalebox{.5}{
    \begin{tikzpicture}
        \triangularGrid{4}{10}{1}{0.865}
        
        \draw[line width=.8mm, gray!70] (0:3) ++ (60:1) --++ (60:3) --++ (120:3) ++ (180:1) --++ (-120:3) --++ (-60:3) ++ (-60:1) ++ (120:2) --++ (60:4) ++ (180:2) --++ (-60:4);

        \draw[line width=.8mm, gray!70] (0:8) ++ (120:1) --++ (120:3) --++(60:2) ;

        \filldraw[line width=.8mm, cyan] (0:2) ++ (60:1) ++ (0:1) ++ (60:1) ++ (180:2) ++ (120:1) --++ (0:1) --++ (0:1) --++ (0:1) ++ (60:1) ++ (180:2) ++ (180:2) ++ (60:1) --++ (0:1) --++ (0:1) --++ (0:1) ++ (120:1) ++ (180:2) ++ (60:1) ++ (0:1) ++ (120:1) ;

        \filldraw[line width=.8mm, blue] (0:2) ++ (60:1) --++ (0:1) ++ (60:3) ++ (120:3) --++ (180:1) ;

        \filldraw[line width=.8mm, blue] (0:8) ++ (120:4) --++ (180:1) ;

        \filldraw[black] (0:3) ++ (120:1) circle(4pt) ++ (0:1) circle(4pt) ++ (60:1) circle(4pt) ++ (180:2) circle(4pt) ++ (120:1) circle(4pt) ++ (0:1) circle(4pt) ++ (0:1) circle(4pt) ++ (0:1) circle(4pt) ++ (60:1) circle(4pt) ++ (180:2) circle(4pt) ++ (180:2) circle(4pt) ++ (60:1) circle(4pt) ++ (0:1) circle(4pt) ++ (0:1) circle(4pt) ++ (0:1) circle(4pt) ++ (120:1) circle(4pt) ++ (180:2) circle(4pt) ++ (60:1) circle(4pt) ++ (0:1) circle(4pt) ;

        \filldraw[black] (0:8) ++ (120:1) circle(4pt) ++ (120:1) circle(4pt) ++ (120:1) circle(4pt) ++ (120:1) circle(4pt) ++ (60:1) circle(4pt) ++ (60:1) circle(4pt) ;
        
    \end{tikzpicture}
    }
    \end{minipage} \hfill
    \begin{minipage}{.49\textwidth}
    \centering
    \vspace{-2.6cm}
    \scalebox{.5}{
    \begin{tikzpicture}
        \triangularGrid{4}{10}{1}{0.865}
        
        \draw[line width=.8mm, gray!70] (0:3) ++ (60:1) --++ (60:3) --++ (120:3) ++ (180:1) --++ (-120:3) --++ (-60:3) ++ (-60:1) ++ (120:2) --++ (60:4) ++ (180:2) --++ (-60:4);

        \draw[line width=.8mm, gray!70] (0:8) ++ (120:1) --++ (120:3) --++(60:2) ;

        \filldraw[line width=.8mm, cyan] (0:2) ++ (60:1) ++ (0:1) ++ (60:1) ++ (180:2) ++ (120:1) --++ (0:1) --++ (0:1) --++ (0:1) ++ (60:1) ++ (180:2) ++ (180:2) ++ (60:1) --++ (0:1) --++ (0:1) --++ (0:1) ++ (120:1) ++ (180:2) ++ (60:1) ++ (0:1) ++ (120:1) ;

        \filldraw[black] (0:3) ++ (120:1) circle(4pt) ++ (0:1) circle(4pt) ++ (60:1) circle(4pt) ++ (180:2) circle(4pt) ++ (120:1) circle(4pt) ++ (0:1) circle(4pt) ++ (0:1) circle(4pt) ++ (0:1) circle(4pt) ++ (60:1) circle(4pt) ++ (180:2) circle(4pt) ++ (180:2) circle(4pt) ++ (60:1) circle(4pt) ++ (0:1) circle(4pt) ++ (0:1) circle(4pt) ++ (0:1) circle(4pt) ++ (120:1) circle(4pt) ++ (180:2) circle(4pt) ++ (60:1) circle(4pt) ++ (0:1) circle(4pt) ;

        \filldraw[black] (0:8) ++ (120:1) circle(4pt) ++ (120:1) circle(4pt) ++ (120:1) circle(4pt) ++ (120:1) circle(4pt) ++ (60:1) circle(4pt) ++ (60:1) circle(4pt) ;

        \draw[rounded corners,dashed,line width=1pt] (0:2.5) ++ (120:.6) --++ (120:3.4) --++ (60:3.4) --++ (0:.8) --++ (-120:1.3) --++ (-60:.8) --++ (0:.7) --++ (60:.7) --++ (120:1.4) --++ (0:.6) --++ (-60:3.4) --++ (-120:3.4) --++ (180:.7) --++ (60:1.3) --++ (120:.8) --++ (180:.6) --++ (-120:.8) --++ (-60:1.3) --++ (180:.7) ;

        \draw[rounded corners,dashed,line width=1pt] (0:7.7) ++ (120:.5) --++ (120:3.6) --++ (60:2.4) --++ (0:.8) --++ (-120:2.4) --++ (-60:3.6) --++ (180:.8); 
    \end{tikzpicture}
    }
    \end{minipage}
    \caption{Example of grey, dark blue and light blue edges in two grey components connected by a DBE. The two grey components without any of the DBEs are traced in the second subfigure.}
    \label{fig:example-all-edges}
\end{figure}

We make the following observations on whether neighbouring particles can detect having common chirality.

\begin{note} \label{obs:grey-edges-orientation-detection}
    Any pair of particles, $\{p, p'\}$, connected by a grey edge can detect whether they have the same chirality. 
\end{note}

 This is done by each particle looking at the port leading to the edge connecting them. If the grey edge that is incident to port $i \in \{1,2,4,5\}$ of $p$ and incident to port $i' \in \{1,2,4,5\}$ of $p'$ satisfies $i + 3 \mod 6 = i'$ and $i' + 3 \mod 6 = i$ (resp. $i + 3 \mod 6 \neq i'$ and $i' + 3 \mod 6 \neq i$) the particles know they agree (resp. do not agree) on chirality.

\begin{note} \label{obs:light-blue-edges-orientation-detection}
    Any pair of particles, $\{p, p'\}$, connected by an LBE can detect whether they have the same chirality.
\end{note}

Each of $p$ and $p'$ is connected by a grey edge to a common occupied neighbour, $q$. As in Note \ref{obs:grey-edges-orientation-detection}, $\{p,q\}$ (resp. $\{p',q\}$) know whether they have common chirality. Therefore, $q$ can also inform $p$ (resp. $p'$) whether it has the same chirality as $p'$ (resp. $p$).

\begin{note} \label{obs:dark-blue-edges-orientation-detection}
    Any pair of particles, $\{p, p'\}$, connected by a DBE cannot locally detect whether they have the same chirality.
\end{note}

Contrary to Note \ref{obs:grey-edges-orientation-detection}, it is always the case that $i + 3 \mod 6 = i'$ and $i' + 3 \mod 6 = i$, so this cannot be used to detect chirality. Contrary to Note \ref{obs:light-blue-edges-orientation-detection}, $p$ and $p'$ do not have common neighours and by the problem definition they cannot move to communicate information like in \cite{emek2019deterministic,di2020shape}. So $p$ and $p'$ cannot locally detect whether they have common chirality.

 The difference in whether particles on horizontal edges can detect having common chirality (Notes \ref{obs:light-blue-edges-orientation-detection} and \ref{obs:dark-blue-edges-orientation-detection}) is why we split horizontal edges into DBEs and LBEs. 

\subsection*{Pseudocode Encoding}
Each time a particle is activated, it performs a (possibly empty) set of actions based on its current state, evaluates the messages it receives and based on its current state and messages it performs another (possibly empty) set of actions and either transitions to a new state or remains in its current state.  The pseudocode of this paper is intended to give a high level description of the detailed algorithm presented in Section \ref{sec:algorithm-description}. In particular we omit all variable updates in the pseudocode and only include the states, transition conditions and actions taken by the particles. Each procedure is a set of states and each state consists of a possibly empty set of actions and a function. The function is of the form \textsc{FunctionName} ( $c_1$: $a_1, \ldots, a_{k}$, \texttt{state$_1$} ; $\ldots$ ; $c_m$: $a_1, \ldots, a_{k'}$, \texttt{state$_m$} ). Each $c_i$ corresponds to a condition and each condition leads to a state. Conditions are evaluated in order and when a condition becomes true, the actions after ``\texttt{:}'' are executed and the particle moves to the new state without evaluating the remaining conditions. If a state is not specified, it is implied that the particle remains in the same state. If no condition is true, the particle remains in the same state until a condition is satisfied. For ease of presentation, the conditions and their corresponding actions and state changes are separated by ``\texttt{;}''. Consecutive actions are separated by commas. At the beginning of each procedure, we list the states included in that procedure and we mark the states a particle may be in when entering the procedure in bold.

\section{Explicit Termination Impossibility} \label{sec:termination-impossibility}
We prove that an explicitly terminating LE algorithm is not possible without imposing any restrictions on the system, using the same basic indistinguishability argument as in \cite{angluin1980local}. Although we assume a synchronous scheduler, the same result holds for sequential activations by activating all equivalent particles (as defined in Figure \ref{fig:s1-and-s2}) sequentially instead of simultaneously. Our proof only addresses the deterministic case, however, it can be shown that no probabilistic algorithm solves the problem by applying the method of \cite{itai1981symmetry}. 

\begin{figure}[th]
    \centering
    \vspace{-4cm}
    \begin{subfigure}[b]{0.49\textwidth}
        \centering
        \begin{tikzpicture}[scale=.5]
            \triangularGrid{5}{4}{1}{0.865}

            \draw[line width=.7mm, blue] (60:1) ++ (0:1) ++ (90:2*.865) --++ (0:1) ++ (60:1) ++ (120:1) --++ (180:1) ++ (-120:1) ;

            \draw[line width=.7mm, gray!70] (60:1) ++ (0:1) ++ (90:2*.865) ++ (0:1) --++ (60:1) --++ (120:1) ++ (180:1) --++ (-120:1) --++ (-60:1) ;
            
            \filldraw (60:1) ++ (0:1) ++ (90:2*.865) circle(5pt) ++ (0:1) circle(5pt) ++ (60:1) circle(5pt) ++ (120:1) circle(5pt) ++ (180:1) circle(5pt) ++ (-120:1) circle(5pt) ;

            \node[scale=.8] (A) at (1.5,2.3*.865) {a};
            \node[scale=.8] (A) at (2.5,2.3*.865) {b};

            \node[scale=.8] (A) at (.5,4*.865)  {f};
            \node[scale=.8] (A) at (3.5,4*.865) {c};

            \node[scale=.8] (A) at (1.5,5.6*.865) {e};
            \node[scale=.8] (A) at (2.5,5.6*.885) {d};
        \end{tikzpicture}
        \caption{\centering $S_1$}
        \label{fig:s1}
    \end{subfigure} \hfill
    \begin{subfigure}[b]{0.49\textwidth}
        \centering
        \begin{tikzpicture}[scale=.5]
            \triangularGrid{5}{4}{1}{0.865}
            \node[] at (2,0) (a) {\huge $\vdots$} ;

            \draw[line width=.7mm, blue] (60:1) ++ (0:1) --++ (0:1) ++ (60:1) ++ (120:1) --++ (180:1) ++ (120:1) ++ (60:1) --++ (0:1) ++ (60:1) ++ (120:1) --++ (180:1) ++ (120:1) ++ (60:1) --++ (0:1) ;

            \draw[line width=.7mm, gray!70] (60:1) ++ (0:1) ++ (0:1) --++ (60:1) --++ (120:1) ++ (180:1) --++ (120:1) --++ (60:1) ++ (0:1) --++ (60:1) --++ (120:1) ++ (180:1) --++ (120:1) --++ (60:1) ++ (0:1) ;
            
            \filldraw (60:1) ++ (0:1) circle(5pt) ++ (0:1) circle(5pt) ++ (60:1) circle(5pt) ++ (120:1) circle(5pt) ++ (180:1) circle(5pt) ++ (120:1) circle(5pt) ++ (60:1) circle(5pt) ++ (0:1) circle(5pt) ++ (60:1) circle(5pt) ++ (120:1) circle(5pt) ++ (180:1) circle(5pt) ++ (120:1) circle(5pt) ++ (60:1) circle(5pt) ++ (0:1) circle(5pt) ; 
            
            \node[] at (2,10.5*0.865) (a) {\huge $\vdots$} ;

            \node[scale=.8] (A) at (.75,.865)   {a$_1$};
            \node[scale=.8] (A) at (.75,5*.865) {a$_2$};
            \node[scale=.8] (A) at (.75,9*.865) {a$_3$};

            \node[scale=.8] (A) at (3.2,.865)   {b$_1$};
            \node[scale=.8] (A) at (3.2,5*.865) {b$_2$};
            \node[scale=.8] (A) at (3.2,9*.865) {b$_3$};

            \node[scale=.8] (A) at (3.75,2*.865) {c$_1$};
            \node[scale=.8] (A) at (3.75,6*.865) {c$_2$};

            \node[scale=.8] (A) at (3.2,3*.865) {d$_1$};
            \node[scale=.8] (A) at (3.2,7*.865) {d$_2$};

            \node[scale=.8] (A) at (.75,3*.865) {e$_1$};
            \node[scale=.8] (A) at (.75,7*.865) {e$_2$};

            \node[scale=.8] (A) at (.4,4*.865) {f$_1$};
            \node[scale=.8] (A) at (.4,8*.865) {f$_2$};
        \end{tikzpicture}
        \caption{\centering Infinite path}
        \label{fig:s2-infinite}
    \end{subfigure}
    \caption{Systems where all equivalent particles have the same local information cannot locally differentiate between the configurations.}
    \label{fig:s1-and-s2}
\end{figure}
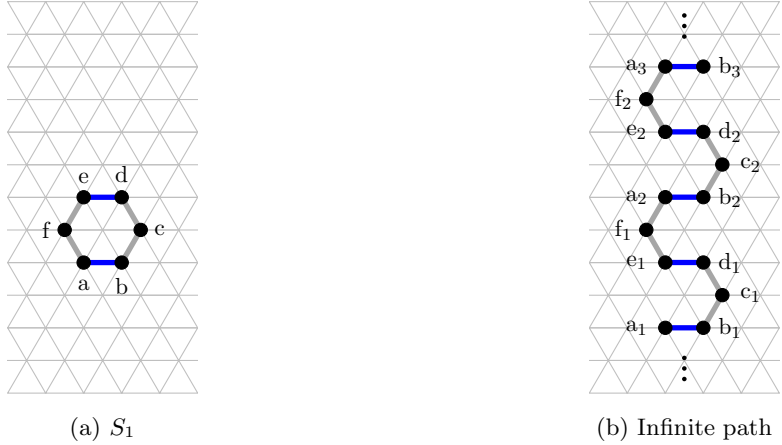

\begin{theorem} \label{th:termination-impossibility}
    There does not exist a terminating algorithm that solves LE under a fair synchronous scheduler when the initial configuration can contain holes and the particles cannot move, even if all particles agree on a common direction.
\end{theorem}
\begin{proof}
     Let $S_1$ be the particle system in Figure \ref{fig:s1} and $S_{\emph{inf}}$ be the infinite particle system presented in Figure \ref{fig:s2-infinite}. Each occupied node $v_i$ in $S_{\emph{inf}}$ is mapped to an occupied node $v$ in $S_1$, such that $v_i \mapsto v$ for $v \in \{a,b,c,d,e,f\}$ and $i \in \mathbb{Z}$. We assume that all particles $b_i, c_i$ or $d_i$ (resp. $a_i, e_i$ or $f_i$) in $S_{\emph{inf}}$ have common (resp. opposite) chirality as the particle they are mapped to in $S_1$. Locally, the neighbourhood at distance one of $\{a_i, b_i, \ldots, f_i\} \in S_{\emph{inf}}$ is the same as the neighbourhood at distance one of $\{a,b,\ldots,f\} \in S_1$, so particles in $S_1$ and $S_{\emph{inf}}$ have the same input and local information. Let us assume that there exists an algorithm $\mathcal{A}$ solving Terminating LE in all connected configurations. Let us further assume that $\mathcal{A}$ solves Terminating LE in $S_1$ after $s$ synchronous steps. Let $u$ be a node in $S_{\emph{inf}}$ and let $S_2$ be the subgraph induced in $S_{\emph{inf}}$ by all particles at distance $k \geq (2*6 + 2) + (2*s + 2)$ from $u$. Now let us consider a synchronous execution of $\mathcal{A}$ in $S_2$. By construction, there must exist at least two particles in $S_2$, $q$ and $q'$, that are in the same state as the elected particle in $S_1$ after $s$ steps of $\mathcal{A}$ such that both $q$ and $q'$ are at distance at least $s+1$ from each endpoint of $S_2$. Therefore, $\mathcal{A}$ elects at least two leaders in $S_2$ and there cannot exist any algorithm solving Terminating LE.
\end{proof}

\section{Message Forwarding on Boundaries} \label{sec:messages-on-boundaries}
Existing algorithms such as \cite{bazzi2019stationary,dufoulon2021efficient} use chirality to move messages along boundaries. Although in this work we assume that particles do not have common chirality, here, we define a method to substitute the chirality assumption in existing algorithms in order to use techniques from previous work. The following procedure forwards messages on boundaries of grey components and we prove that a message originating at a particle in some boundary $B$ is received by particles on $B$ and it is not received by particles not on $B$. 

\begin{note} \label{lem:light-blue-neighbours-on-one-boundary}
    Two particles, $p$ and $p'$, that are connected by an LBE are neighbours in exactly one boundary, even if they both participate in more than one common boundary.
\end{note}
    The edge connecting $p$ and $p'$ can only be an LBE if $p$ and $p'$ are connected in the horizontal direction and have a common neighbouring node occupied by a particle. Furthermore, $p$ and $p'$ can only be neighbours on a boundary if they share at least one unoccupied neighbour. Therefore, $p$ and $p'$ must have exactly one occupied and one unoccupied common neighbour which means that they are neighbours in exactly one boundary.

\begin{note}\label{lem:grey-components-two-boundaries}
    Each particle belonging to a grey component $G$, is in at most two boundaries in $G$.
\end{note}
    Notice that any particle that is on three boundaries must have either a neighbour reached through port 0 or through port 3. Let $p$ be a particle that is on three boundaries and $p'$ be a neighbour of $p$, connected to $p$ on the horizontal direction. Then $p'$ cannot share any neighbour, $p_n$, with $p$, otherwise, $p$ is not on the common boundary of $p'$ and $p_{n}$. So the edge connecting $p$ and $p'$ must be dark blue. Therefore, $p$ cannot be on three boundaries without a DBE and consequently particles in grey components are on at most two boundaries.

\underline{Message Forwarding on Boundaries:} Let $p_s$, $p_1$ and $p_2$ be three consecutive particles on boundary $B$ and say that $p_s$ holds a message it will send to $p_1$. Call the common non--particle neighbour of $p_s$ and $p_1$ on $B$, $o_1$ and the common non--particle neighbour of $p_1, p_2$ on $B$, $o_2$. When $p_s$ sends the message to $p_1$ it attaches the label of the port of $p_1$ leading to $o_1$ using the orientation of $p_1$. Call that label \emph{boundary--label}. When $p_1$ receives the message from $p_s$ through port $z$ it reads the \emph{boundary--label} attached to the message and sets \textit{next--particle} to be the first particle to be reached by following the cyclic port ordering $<z, \emph{boundary--label}, \ldots>$, which in this case we have called $p_2$. $p_1$ detects whether it has common chirality with $p_2$, using the method described in Note \ref{obs:grey-edges-orientation-detection} if $(p_1,p_2)$ is a grey edge or the method described in Note \ref{obs:light-blue-edges-orientation-detection} if $(p_1,p_2)$ is an LBE. Call the port of $p_1$ leading to $p_2$ port $x$, the port of $p_2$ leading to $p_1$ port $y$, the port of $p_1$ leading to $o_2$ port $i$ and the port of $p_2$ leading to $o_2$ port $j$. If $p_1$ and $p_2$ have common chirality and $i = x-1 \mod 6$ (resp. $i = x + 1 \mod 6$), $p_1$ calculates $j = y + 1 \mod 6$ (resp. $j = y-1 \mod 6$). Equivalently, if $p_1$ and $p_2$ do not have common chirality and $i = x-1 \mod 6$ (resp. $i = x + 1 \mod 6$), $p_1$ calculates $j = y - 1 \mod 6$ (resp. $j = y+1 \mod 6$). Finally, $p_1$ sets \emph{boundary--label} to $j$ and sends the message to $p_2$. 

\begin{theorem} \label{th:correct-boundary-communication}
    Let $\{p,p'\}$ be a pair of particles neighbouring on at least one boundary of a grey component. If both $p$ and $p'$ use the above procedure, any message passing from $p$ to $p'$ always remains on the same boundary. 
\end{theorem}
\begin{proof}
    Let $p$ be the particle from which $p'$ receives the message. If $p$ is the particle the message originates from, $p$ decides the boundary, $B$, on which the message is sent. From the definition of a \textit{local boundary}, we know that $p$ and $p'$ share a common non--particle neighbour, say $o$. Due to the common direction, if $p$ is connected to $p'$ through a right port (resp. left port), $p'$ is connected to $p$ through a left port (resp. right port). Furthermore, if $(p,p')$ is grey and $p$ is connected to $o$ through a left port (resp. right port), $p'$ is also connected to $o$ through a left port (resp. right port). Particles connected by an LBE know they only have one common non--particle neighbour. Since $o$ is a common neighbour it must be reached by a port at distance one from the port connecting $p$ to $p'$. The type of edge (i.e., grey or LBE), whether $o$ is reached by a left or right port and the port of $p'$ leading to $p$ are known to $p$. Moreover, the handedness of $p'$ can be computed by Notes \ref{obs:grey-edges-orientation-detection} and \ref{obs:light-blue-edges-orientation-detection}. Therefore, $p$ can compute the port of $p'$ leading to $o$. If $p$ is not the particle the message originates from, $p$ learns the boundary the message is forwarded on by receiving the port label leading to the common non--particle neighbour defining the boundary when receiving $m$ by its predecessor on the boundary, in the same way as $p'$ learns the boundary in the case where the message originates from $p$ and the proof remains the same. 
\end{proof}

\section{Algorithm Description} \label{sec:algorithm-description}
The system consists of a number of grey components connected by DBEs. Each grey component elects a local leader. In each round, every component attempts to \textit{``compete''} with two neighbouring components, where \textit{``neighbouring''} refers to components connected by a DBE. The result of a competition for each component is \textit{win}, \textit{lose} or \textit{draw}. The first two results occur when the competing components detect a difference between them. In this case, the components merge creating a new component consisting of all the particles of the two previous components and the local leader of the winning component. The result of a competition is a draw when two components cannot detect a difference between them, for example when a component is competing with itself. In this case, the components remain separate and do not compete on that DBE again, unless at least one of them merges with another component. We prove that when no competition on any DBE is possible anymore, all particles belong to one common component with one leader. We address each part of the algorithm separately.

\subsection{Grey Leader Election} \label{sec:grey-LE}
We begin by electing a unique leader in every grey component. Each particle sees its neighbourhood at distance one and differentiates between grey, light blue and dark blue edges. For a DBE $uv$ connecting two particles, $p$ at $u$ and $p'$ at $v$, $p$ (resp. $p'$) removes $p'$ (resp. $p$) from its list of neighbours during Grey LE and treats $v$ (resp. $u$) as a node not occupied by a particle. In existing stationary LE algorithms (e.g., \cite{bazzi2019stationary}) and algorithms that present intermediate tools for stationary particles (e.g., \cite{dufoulon2021efficient,emek2019deterministic}), it is assumed that particles have common chirality which is used to exchange information along boundaries. For $k$--symmetric configurations, those algorithms elect $k \in \{1,2,3,6\}$ leaders in symmetric positions on the outer boundary. To facilitate the presentation of our algorithm, we say that we use the algorithm presented in \cite{bazzi2019stationary} in this step, however, any algorithm for stationary particles that terminates electing 1,2,3 or 6 heads in symmetric positions on the outer boundary could be used. In \cite{bazzi2019stationary} common chirality is used to move information along boundaries. We substitute the assumption of common chirality in \cite{bazzi2019stationary} in the parts of the system where this is possible (i.e., grey components) by the method in Section \ref{sec:messages-on-boundaries}. However, since the particles do not agree on chirality we execute the algorithm of \cite{bazzi2019stationary} in both directions on each boundary of each grey component simultaneously, electing $k \in \{1,2,3,6\}$ heads that know the number of elected particles for each direction of the outer boundary. Particles in inner boundary competitions learn that they are not elected when the competition on the unique outer boundary terminates and the procedure described in Section \ref{sec:tree-construction} begins. We do not repeat the pseudocode for the algorithm that is presented in \cite{bazzi2019stationary} but we assume that after its execution, every particle is either in state \texttt{Head} (i.e., a leader elected by the algorithm in \cite{bazzi2019stationary}) or \texttt{NotHead} (i.e., a particle which is not elected by the algorithm in \cite{bazzi2019stationary}). 

Each head calculates an ID for itself and a local ID for the $k - 1$ heads on the boundary for each direction, $dir$, it is elected in. The ID is the port number $a \in \{0,\ldots,5\}$ leading to the particle's neighbour on the boundary in $dir$. The calculated ID for the $i$--th head, $i = \{1,\ldots,k-1\}$, on the boundary is \emph{ID}$_i = a + i*\frac{6}{k} \mod 6$, using the local port numbering. Since the heads do not actually have common chirality, different heads may compute distinct IDs for the same head. However, since directions 0 and 3 are common, if a head $h$ computes 0 or 3 for its own ID, all other heads on the boundary also compute 0 or 3 respectively as the ID of $h$.  We impose the ordering $\emph{ID} = 0 > \emph{ID} = 3 > \emph{ID} = \{1,5\} > \emph{ID} = \{2,4\}$. A particle remains a head if its ID is the maximum with respect to the given ordering among the IDs it has locally computed. That is, if there exists a head with $\emph{ID} = 0$, that particle remains a head and all other heads withdraw. If no such particle exists, the particle with $\emph{ID}=3$ remains a head and so on. Eventually, two heads remain per boundary after winning the $dir$ and the $dir'$ competitions.

The two remaining heads compete so that only one head remains on the boundary. Each of the heads sends a competition message in the direction it was elected. Either both messages are received by a common particle, $p_c$, or by two neighbouring particles $p_{c_1}, p_{c_2}$. In the first case, the single particle picks one of the heads based on its local orientation. In the latter case, $p_{c_1}, p_{c_2}$ know which is the rightmost particle by looking at the ports connecting them. The rightmost particle then sends an \textit{elected} message to the head that reached it and the leftmost particle sends a \textit{not--elected} message to the remaining head. After this step, one leader remains and it is on the outer boundary. The leader then initiates the \textit{Tree Construction} phase of the algorithm and particles located on inner boundaries learn that LE in the component has terminated. This algorithm is also presented in Procedure \ref{proc:grey-LE}. Notice that after the execution of Procedure \ref{proc:grey-LE}, each particle on the outer boundary is either in state \texttt{Leader} or in state \texttt{Follower}. 

\begin{procedure}[!h]
    \caption{GreyLE()}
    \label{proc:grey-LE}

    States: \{\textbf{\small Head}, \textbf{\small NotHead}, MovingInformation, CompetingHead\} \\

    \underline{In state \texttt{Head}:} \\
    \quad Calculate own ID \\
    \quad Calculate ID of other heads \\
    \quad \textsc{ElectionPerDirection} ( \textit{only head in direction}: \texttt{CompetingHead}; \\
    \hspace{5.1cm} \textit{own ID \hspace{.1cm}  = 0}: \texttt{CompetingHead}; \\ 
    \hspace{5.1cm} \textit{other ID              = 0}: \texttt{NotHead};\\ 
    \hspace{5.1cm} \textit{own ID  \hspace{.1cm} = 3}: \texttt{CompetingHead}; \\
    \hspace{5.1cm} \textit{other ID              = 3}: \texttt{NotHead}; \\
    \hspace{5.1cm} \textit{own ID \hspace{.05cm} = 1 or 5}: \texttt{CompetingHead}; \\
    \hspace{5.1cm} \textit{own ID \hspace{.05cm} = 2 or 4}: \texttt{NotHead} ) \\

    \underline{In state \texttt{NotHead}:} \\
    \quad \textsc{ElectionPerBoundary} ( \\
    
    \qquad \textit{receive competition--message from dir \emph{AND} \\
    \qquad not receive competition--message from dir$'$ \emph{AND} \\ \qquad neighbour does not receive competition--message from dir$'$}: \\
    \hspace{2cm} forward competition--message to dir, \texttt{MovingInformation};\\
    
    \qquad \textit{receive competition--message from dir \emph{AND} \\
    \qquad receive competition--message from dir$'$}: \\
    \hspace{2cm} send leader--message in dir$'$ and follower--message in dir, \\
    \hspace{2.1cm}\texttt{Follower}; \\
    
    \qquad \textit{receive competition--message from dir \emph{AND} \\
    \qquad neighbour receives competition--message from dir$'$ \emph{AND} \\
    \qquad sees neighbour through a right port}: \\
    \hspace{2cm} send follower--message to dir$'$, \texttt{Follower}; \\ 

    \qquad \textit{receive competition--message from dir \emph{AND} \\ \qquad neighbour receives competition--message from dir$'$ \emph{AND} \\
    \qquad sees neighbour through a left port}: \\
    \hspace{2cm} send leader--message to dir$'$, \texttt{Follower} ) \\

    \underline{In state \texttt{MovingInformation}:} \\
    \quad \textsc{ElectionPerBoundary} ( \\
    \qquad \textit{receive competition--message from dir$'$}: \\
    \hspace{2cm} send competition message to dir, \texttt{Follower} ) \\

    \underline{In state \texttt{CompetingHead}:} \\
    \quad send competition--message in the direction of election\\
    \quad \textsc{ElectionPerBoundary} ( \\
    \qquad \textit{receive leader--message}: \texttt{Leader}; \\
    \qquad \textit{receive follower--message}: \texttt{Follower} ) \\
    
\end{procedure}

\subsubsection*{Correctness}
We prove the correctness of Grey Leader Election. 

\begin{lemma} \label{lem:LE-on-boundary}
    After executing the LE algorithm of \cite{bazzi2019stationary} in both directions for every boundary, 1,2,3 or 6 leaders are elected on symmetric positions in the outer boundary for each direction and each head knows the direction of the competition in which it was elected and the number $k$ of elected heads in that direction. 
\end{lemma}
\begin{proof}
    The fact that 1,2,3 or 6 leaders are elected comes directly from the correctness of the LE algorithm in \cite{bazzi2019stationary}. The chirality assumption of \cite{bazzi2019stationary} is substituted by Theorem \ref{th:correct-boundary-communication} which proves that particles can consistently exchange information along boundaries of grey components. From the termination conditions in \cite{bazzi2019stationary}, we know that the algorithm stops when the symmetry between the heads on the outer boundary cannot be broken, that the algorithm does not terminate in inner boundaries and that the number of heads on the outer boundary is $k \in \{1,2,3,6\}$ depending on the symmetry of the system. Furthermore, the direction of the competition is known by construction of Grey Leader Election, therefore all conditions of the lemma statement are satisfied.
\end{proof}

Each head on the outer boundary can compute a set of distinct IDs, each corresponding to another head on the boundary, with respect to its own coordinate system, even in the setting of \cite{bazzi2019stationary}. However, in that case it is possible that different particles locally compute the same ID for distinct heads in the boundary. We show that by adding agreement on one direction, the heads do not compute the same set of IDs corresponding to different particles. 

\begin{lemma} \label{lem:unique-head-in-dir}
    Using the locally computed identifiers of all heads on the boundary for a given direction, all heads choose the same particle to remain a head, without communicating with each other. Eventually, only one head remains in each direction on the outer boundary. 
\end{lemma}
\begin{proof}
    Due to the particles not having chirality, the sets of IDs computed by different heads may be different. The calculation of 0 or 3 as the ID for any head is consistent among all heads since 0 and 3 are common directions. Suppose that for $k \in \{3,6\}$ there does not exist a particle with $\emph{ID} = 0$ or $\emph{ID} = 3$. Using the formula for calculating IDs, \emph{ID}$_i = a + i*\frac{6}{k} \mod 6$, the ID calculated for at least one of the heads is either 0 or 3 and we immediately get a contradiction. For $k=2$ it is possible that no head with $\emph{ID} \in \{0,3\}$ exists. This time, we use the fact that port 1 is mapped either to port 1 or port 5 (similarly for port 5), depending on the chirality of the particles and the same holds for ports 2 and 4. A head, $h$, calculating $a_h \in \{1,5\}$ (resp. $a_h \in \{2,4\}$) for its ID calculates $a_{h'} \in \{2,4\}$ (resp. $a_{h'} \in \{1,5\}$) for the other head, $h'$. In the same setting $h'$ computes $a_{h'} \in \{2,4\}$ (resp. $a_{h'} \in \{1,5\}$) for its own ID and $a_h \in \{1,5\}$ (resp. $a_h \in \{2,4\}$) for the ID of $h$. Thus $a \in \{1,5\}$ corresponds to only one of the two heads, and it is elected. Therefore, in all cases one head remains on the outer boundary for each direction in each grey component.
\end{proof}

\begin{lemma} \label{lem:unique-boundary-leader}
    For a boundary $B$ where exactly one leader remains for each direction, the leaders can always reach either one common node or a pair of neighbouring nodes occupied by particles. The reached particle(s) can differentiate between the leaders and elect a unique leader for the boundary.
\end{lemma}
\begin{proof}
    Since the heads send \textit{competition messages} in opposite directions on the boundary, the messages must eventually reach either a common particle or two neighbouring particles. If the comparison messages are held by a single particle, that particle decides the leader. Otherwise, the comparison messages are held by adjacent particles. Due to the triangular grid, one of the particles is to the right of the other (either horizontally or diagonally) and the particles can locally decide which one is the rightmost in the pair due to the port numbers. Since the rightmost particle decides the leader, a unique leader is elected.
\end{proof}

From the above lemmas, we obtain the following.

\begin{theorem} \label{th:grey-leader-election-correct}
    After the execution of the Grey Leader Election algorithm, a unique leader is elected on the outer boundary of every grey component.
\end{theorem}

\subsection{Tree Construction} \label{sec:tree-construction}
We encode each grey component so that leaders in different grey components, connected by DBEs, can compare their components. We create a spanning tree of the component, using a standard technique from the distributed computing literature, in which the root is the unique leader of the component. Each particle in the tree computes a label encoding its neighbourhood. By following the reasoning in \cite{di2020shape}, we encode the neighbourhood of a particle $p$ using six characters, each corresponding to the node reachable through port $i \in \{0,\ldots,5\}$. Each character in the label has one of the following values. 
\begin{itemize}
    \item \textit{P}: Node occupied by the parent of $p$.
    \item \textit{C}: Node occupied by a child of $p$.
    \item \textit{D}: Node connected to $p$ through a DBE. Note that if port[0] is encoded as D the DBE is outgoing while if port[3] is encoded as D the DBE is incoming. 
    \item \textit{E}: Node not occupied by a particle.
    \item \textit{N}: Neighbour particle of $p$ that is not a neighbour on the tree or a DBE.
\end{itemize}
Each particle knows its parent and its children, as well as the information that is locally observed from its immediate neighbourhood (i.e., non--particle or particle nodes and incoming or outgoing DBEs). Notice that the characters in the label are mutually exclusive during Tree Construction, so it cannot be the case that more than one characters in \{P,C,D,E,N\} are needed to represent a neighbour of $p$. We do not repeat the pseudocode which can be found in \cite{di2020shape}, however, we assume that after the execution of this part of the algorithm, each particle is in state \texttt{OutgoingDBE}. The correctness of the tree construction phase comes directly from \cite{di2020shape}.

\begin{procedure}[h!]
    \caption{TreeConstruction()}
    \label{proc:tree-construction}

    States: \{\textbf{\small Leader, Follower}\} \\

    \tcc{Execute the tree construction procedure of \cite{di2020shape} and switch to state OutgoingDBE in Procedure \textsc{CompetitionEdges}}
\end{procedure}

\subsection{Component Competition} \label{sec:component-competition}
In this section, we describe how particles in different components are compared. We begin by describing how the trees defined in Section \ref{sec:tree-construction} are traversed in order to be compared to adjacent trees through the directed DBEs connecting them. Then we discuss how components choose which DBEs to compete on and how components merge after a comparison, if merging is possible. Finally, we show that comparisons eventually stop and that only one leader remains in the system when no more comparisons are made. We refer to each iteration of the set of procedures in this section (i.e., Tree Traversal, Tree Comparisons and Merging Components) by a component as a \textit{round}.

\begin{figure}[htpb]
    \centering
    \vspace{-.5cm}
    \scalebox{.8}{
    \begin{tikzpicture}[->,>=stealth',every node/.style={circle,draw},level 1/.style={sibling distance=35mm},level 2/.style={sibling distance=15mm}
    ]
    \node[label={[xshift=-.5em, yshift=0em] 1}, label={[xshift=0em, yshift=-2.5em] 7}, label={[xshift=.5em, yshift=0em] 9}] (nA) {}
       child { node[label={[xshift=-1em, yshift=-1.3em] 2}, label={[xshift=0em, yshift=-2.5em] 4}, label={[xshift=1em, yshift=-1.3em] 6}] (nB) {}
                  child { node[label={[xshift=-1em, yshift=-1.3em] 3}] (nC) {} }
                  child { node[label={[xshift=-1.5em, yshift=-1.3em] 5}] (nD) {} }
                }
       child { node[label={[xshift=-1.5em, yshift=-1em] 8}] (nE) {} };
    
      \draw[->,black!70,rounded corners,dashed,line width=0.7pt]
        ($(nA) + (-0.4,0.2)$) --
        ($(nB) +(-0.3,0.4)$) --
        ($(nB) +(-0.6,0.0)$) --
        ($(nC)  +(-0.4,0.3)$) --
        ($(nC)  +(-0.5,0.0)$) -- 
        ($(nC)  +(-0.4,-0.35)$) --
        ($(nC)  +(0.0,-0.5)$) --
        ($(nC)  +(0.4,-0.35)$) --
        ($(nC)  +(0.5,0.0)$) --
        ($(nB)  +(0.0,-0.4)$) --
        ($(nD)  +(-0.4,0.3)$) --
        ($(nD)  +(-0.3,0.0)$) -- 
        ($(nD)  +(-0.2,-0.35)$) --
        ($(nD)  +(0.0,-0.5)$) --
        ($(nD)  +(0.4,-0.35)$) --
        ($(nD)  +(0.5,0.0)$) --
        ($(nD)  +(0.4,0.2)$) --
        ($(nB)  +(0.4,0.0)$) --
        ($(nA)  +(0.0,-0.4)$) --
        ($(nE)  +(-0.6,0.0)$) --
        ($(nE)  +(-0.4,-0.35)$) --
        ($(nE)  +(0.0,-0.5)$) --
        ($(nE)  +(0.4,-0.35)$) --
        ($(nE)  +(0.5,0.2)$) --
        ($(nE) +(0.3,0.4)$) --
        ($(nA) + (0.4,0.2)$)
        ;
    \end{tikzpicture}
}
\caption{An example of cyclic--DFS demonstrating the traversal of a tree represented as a ring. The numbers denote the order in which each node is visited.}
\label{fig:tree-traversal}
\end{figure}
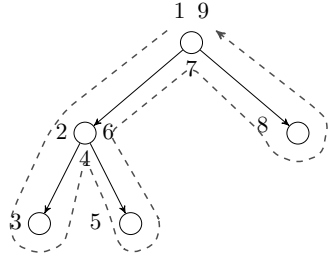

\underline{\textbf{Tree Traversal:}} Each node orders its children in increasing order with respect to its port numbers so that the ordering shown in Figure \ref{fig:tree-traversal} can be obtained. We call the ring calculated in such a way ``cyclic--DFS'', also known as the Euler tour of the tree. Notice that the five node tree of Figure \ref{fig:tree-traversal} is represented by a nine--node ring and that in the traversal order of Figure \ref{fig:tree-traversal}, the root is the first as well as the last node. The root marks its first label as \textit{root--label} and its last label as the \textit{last--node}. We later use the ring encoding to move information in the form of labels and tokens between particles of the same tree. 

\underline{Encoding a tree as a ring:} Each node of the tree is split into $c + 1$ virtual nodes called \emph{agents}, where $c$ is the number of children of that node in the tree. Each agent represents one node of the equivalent ring. Let $root = v_1,v_2,\ldots,v_n$ be the nodes of the tree encoding of the component. We write $v_i.agent_j$ for $i \in \{1,\ldots,n\}$ and $j \in \{1,\ldots,7\}$ to describe the $j$-th time node $v_i$ is visited in cyclic DFS. Equivalently, $v_i.agent_j$ also denotes that $j-1$ children of $v_i$ have been visited in the traversal. For a node of the tree, $v$, we denote the parent of $v$ by $v^-$ and we write $v^+_{1},\ldots,v^+_{z}$ to represent each of the $1 \leq z \leq 6$ children of $v$. We define pre--node and next--node in terms of nodes of the tree and agents in the following table. In ``Label Forwarding Cycle'' we show how to move information using the ring. The ``Label Forwarding Cycle'' procedure can be started by any arbitrary agent, say $p.agent$. The procedure finishes when the label of a target node, call it target--label, reaches $p.agent$.

\begin{table}[ht]
\begin{center}
\vspace{-.2cm}
\scalebox{1}{
\renewcommand{\arraystretch}{1.2}
\begin{tabular}{ |c|c|c|c| } 
 \hline
 node & agent & pre--node & next--node \\ 
 \hline
 root & 1 & $v_1.agent_{last}$ & $v^+_{1}.agent_1$ \\ 
 root & last & $v^+_{z}.agent_{last}$ & $v_1.agent_1$  \\
 root & $1 < y < \emph{last}$ & $v^+_{y-1}.agent_{last}$ &  $v^+_{y}.agent_{1}$ \\
 internal, $x$-th child of $v^-$ & 1 & $v^-.agent_{x}$  & $v^+_{1}.agent_{1}$ \\
 internal, $x$-th child of $v^-$ & last & $v^+_{z}.agent_{last}$ & $v^-.agent_{x}$ \\
 internal, $x$-th child of $v^-$ & $1 < y < \emph{last}$ & $v^+_{y-1}.agent_{last}$ & $v^+_{y}.agent_{1}$\\
 leaf, $x$-th child of $v^-$ & unique & $v^-.agent_x$ & $v^-.agent_{x+1}$ \\
 \hline
\end{tabular}
}
\caption{Transforming a tree into a virtual ring}
\end{center}
\vspace{-1.2cm}
\end{table}

\underline{Label Forwarding Cycle (LFC):} $p.agent$ simultaneously sends its label to the \textit{pre--node} on the ring and pulls the label, $l$, of the \textit{next--node}, on the ring. When some $p'.agent \neq p.agent$ holds two labels, it forwards the oldest label to the pre--node. When $p'.agent$ holds one label but the pre--node does not hold any label, $p'.agent$ forwards the label it holds to the pre--node. In all other cases, particles do nothing. While $l$ is not the \textit{target--label}, $p.agent$ waits until the next--node gets a new label and then pulls the label from the next--node and forwards its current label to the pre--node. When $l$ is the target--label, $p.agent$ sends a termination--message around the tree through its pre--node. Each particle that receives the termination--message, forwards it if it holds one label or holds it while it holds more than one labels. When the termination--message returns to $p.agent$, the label forwarding procedure finishes. This is described in Procedure \ref{proc:lfc}.

\begin{procedure}[ht]
    \caption{LFC()}
    \label{proc:lfc}

    States: \{\textbf{\small Initiator, NotInitiator}, Termination\}\\
    \underline{In state: \texttt{Initiator}}\\
        \quad \textsc{MoveLabels} ( 
        \textit{received \emph{root--label}}: \texttt{Termination}; \\ 
        \hspace{3.2cm}\textit{next--node holds a label}: send own label to  pre--node and pull \\
        \hspace{7.2cm} label from next--node ) \\
        
    \underline{In state: \texttt{NotInitiator}} \\
        \quad \textsc{MoveLabels} ( \textit{hold two labels}: send oldest label to pre--node; \label{line:label-ordering} 
        \\
        \hspace{3.05cm} \textit{pre--node has no label}: send own label to pre--node; \\
        \hspace{3.05cm} \textit{received termination--message}: \texttt{Termination} ) \\
        
    \underline{In state: \texttt{Termination}} \\
        \quad \textsc{MoveLabels} (\textit{ hold \emph{root--label} and \emph{termination--message}}: \texttt{Ready}; \\
        \hspace{3.05cm} \textit{hold \emph{root--label}}: send termination--message to pre--node;\\
        \hspace{3.05cm} \textit{hold termination--message}: send termination--message to \\
        \hspace{7.65cm} pre--node, \texttt{NotActive}  ) \\
\end{procedure}

The label of an agent $r$ at distance $i$ from $p.agent$ in the ring reaches $p.agent$ once each particle between $r$ and $p.agent$ has performed $2i$ steps. This is because $i-1$ particles between $r$ and $p.agent$ need to forward their label before $r$ sends its own and $i$ more steps are needed to reach $p.agent$.

\underline{\textbf{Tree Comparisons:}} Let $C$ be a grey component. Initially, all incoming and outgoing DBEs are marked as \textit{not--compared}. $C$ marks as \textit{chosen} the first outgoing DBE $uv$ encountered in the cyclic--DFS traversal of the tree encoding $C$ that is not already marked as \textit{compared} if one exists, such that $u \in C$ and $v \in C'$, where $C'$ is also a component. Of the incoming DBEs that are marked as chosen by components at distance one, $C$ marks as \textit{chosen} the first one that is encountered in the cyclic--DFS traversal of the tree encoding it, if one exists. If there is at least one incoming edge, $C$ makes a comparison on the incoming edge it selected. If the outgoing edge selected by $C$ is also selected by a neighboring component as an incoming edge, $C$ makes a comparison on this edge. Thus, $C$ performs at most two comparisons at the same time. It is possible that no outgoing or incoming DBEs exist or that all outgoing DBEs are already marked as \textit{compared}, so $C$ does not make any comparisons. When $C$ selects an incoming DBE for a comparison it messages all other incoming DBEs selected by neighbouring components that they do not participate in the current competition and the corresponding edges remain \textit{not--compared}. The two endpoints of each edge on which a comparison is performed also encode being part of an edge participating in a comparison as part of their label. For example a label ECEDEP is changed to ECE\textbf{\underline{D}}EP as long as the DBE is used for a comparison. When the edge is no longer used, the label is changed back to ECEDEP. The selection of competition edges is also described in Procedure \ref{proc:competition-edges}, with respect to components. 

\begin{procedure}[ht]
    \caption{CompetitionEdges()}
    \label{proc:competition-edges}

    States: \{\textbf{\small OutgoingEdge}, IncomingEdge, BeginCompetition\} \\

    \underline{In state: \texttt{OutgoingEdge}} \\
    \quad \textsc{ChoosingOutgoingDBE} ( \textit{no chosen outgoing DBE}: mark the first outgoing \\ 
    \hspace{2cm} DBE that is \emph{not--compared} in cyclic--DFS as \emph{chosen}, \texttt{IncomingEdge} ) \\ 

    \underline{In state: \texttt{IncomingEdge}} \\
    \quad \textsc{ChoosingIncomingDBE} ( \textit{no chosen incoming DBE}: mark the first incoming \\
    \hspace{2cm} DBE that is \emph{chosen} in cyclic--DFS as \emph{active}, \texttt{BeginCompetition} ) \\ 

    \underline{In state: \texttt{BeginCompetition}} \\
    \quad \textsc{ChoosingCompetitionEdges} ( \\
    \qquad \textit{endpoint to the \emph{active} incoming DBE in cyclic--DFS}: \texttt{Initiator}; \label{line:incoming-DBE} \\
    \qquad \textit{\emph{chosen} outgoing DBE was marked as \emph{active} by neighbouring component}: \label{line:outgoing-DBE} \\
    \hspace{2cm} the endpoint of the chosen outgoing DBE goes to state \texttt{Initiator}, \\
    \hspace{2cm} all other particles go to state \texttt{NotInitiator} ) 
\end{procedure}

Let $p_{db}$ be the endpoint of a DBE selected for a competition. $p_{db}.agent_1$ initiates an LFC, using \textit{root--label} as the target label. If two comparisons are made simultaneously, an LFC is done separately but in parallel for the incoming and the outgoing comparisons. Each node in the tree holds two copies of the labels, one for each of the forwarding procedures. When initiating the LFC, $p_{db}.agent_1$ marks whether the forwarding cycle it initiates is for the comparison on an incoming DBE or an outgoing DBE. When the LFC finishes, $p_{db}$ marks itself as \textit{ready} and waits for the remaining endpoint of the DBE it is competing on, $p'_{db}$, to also mark itself as \textit{ready}. When $p_{db}$ and $p_{db}'$ are both marked as \textit{ready}, they each restart the LFC mechanism and compare each of the labels they receive during LFC to the label received by the other endpoint of the DBE, as follows. If the trees differ in the label of some particle, for the first pair of particles with different labels, take the first position of the label in which a difference is detected, $j$. We arbitrarily define the ordering $P > C > D > \textbf{\underline{D}} > E > N$. We say that the component that has the smaller symbol in position $j$, \textit{loses} the comparison, or equivalently, that it is a \textit{losing} component. If the trees of the compared components have the same labels during the comparison but one tree has more nodes, the component with less particles \textit{loses}. Each competing component waits until all comparisons in which it participates finish before participating in a merge. Every node in the tree knows the number of comparisons the component participates in due to the LFCs. When a particle incident to a DBE finishes a comparison, it sends a message to the root through its parent with the result of the comparison. This is presented in Procedure \ref{proc:tree-comparisons}. A component, $C$, that loses one comparison to another component $C'$, merges with $C'$ using the procedure in ``Merging Components''. If $C$ loses two comparisons, it picks the component reached by the outgoing DBE it lost to, call it $C'$, and only merges with $C'$ using the procedure in ``Merging Components''. Each edge on which either a merge is performed or the components are detected to be equivalent is marked as \textit{compared}. When two components merge, the DBE on which the merge was performed becomes \textit{compared} and all DBEs that were marked as compared in a previous round become \textit{not--compared}. When an endpoint of an edge is marked as \textit{not--compared}, the remaining endpoint of the edge also marks the edge as \textit{not--compared}. 

\begin{procedure}[ht]
    \caption{TreeComparisons()}
    \label{proc:tree-comparisons}

    States: \{\textbf{\small Ready}, Compete, \textbf{\small NotActive}\} \\

    \underline{In state \texttt{Ready}:} \label{line:ready} \\
    \quad \textsc{SynchronizeOnDBE} ( \textit{neighbour on DBE is in state} \texttt{Ready}: \texttt{Compete}) \label{line:neigh-ready} \\

    \underline{In state: \texttt{NotActive}} \\
        \quad \textsc{MoveLabels} ( \\
        \qquad \textit{hold two labels}: send oldest label to previous--node; \\
        \qquad \textit{previous--node does not hold a label}: send own label to previous--node; \\
        \qquad \textit{receive \emph{participating--in--merge} message}: \texttt{InMerge}; \\
        \qquad \textit{receive \emph{merge--complete} message}: \texttt{MergeComplete} ) \\

    \underline{In state \texttt{Compete}:} \\
    \quad \textsc{CompareComponents} (\\
    \qquad \textit{next--node holds a label}: send own label to previous--node and pull label from \\
    \hspace{4.8cm} next--node; \\
    \qquad \textit{own label $>$ competing neighbour label}: \label{line:comp-start} \\
    \hspace{7.4cm} \texttt{WinningMerge}; \\
    \qquad \textit{own label $<$ competing neighbour label}:\\
    \hspace{7.4cm} \texttt{LosingMerge}; \\
    \qquad \textit{own label = competing neighbour label \emph{AND} \\
    \qquad own label is last}: \\
    \hspace{7.4cm} \texttt{LosingMerge}; \\
    \qquad \textit{own label = competing neighbour label \emph{AND} \\
    \qquad competing neighbour label is last}:\\
    \hspace{7.4cm} \texttt{WinningMerge}; \label{line:comp-end} \\
    \qquad \textit{own label = competing neighbour label \emph{AND} \label{line:cont-compete}\\
    \qquad own label is not last \emph{AND}\\
    \qquad competing neighbour label is not last}: \\
    \hspace{7.4cm} \texttt{Compete} );
\end{procedure}

\underline{\textbf{Merging Components:}} After a component $C'$ loses, it merges with the component that won, $C$. Let $uv$, such that $u \in C$ and $v \in C'$, be the DBE on which the merge is executed. The particle at $v$ initiates the merge by sending a message to its parent. The merging message is moved through the tree by each particle marking itself as in--merge and forwarding the message to its parent until the root is reached. When the root receives the merging message, it sets its state to \textit{follower}, marks its in--merge neighbour as its parent by switching its corresponding port in the label to $P$ and sends the merging message back to the in--merge neighbour. Each in--merge particle that receives the merge--message for a second time, marks the node from which it received the message as its child (i.e., switches $P$ to $C$ in its label), marks its remaining in--merge neighbour as its parent by switching its corresponding port in the label to $P$, forwards the merge--message to its in--merge neighbour and unmarks itself as in--merge. When $v$ receives back the merge--message it marks the node from which it received the message as its child and $u$ as its parent in its label. Finally, when $v$ changes its label to mark $u$ as a parent (i.e., $v$ changes \underline{\textbf{D}} in its label to $P$), $u$ also marks $v$ as its child (i.e., $u$ changes \underline{\textbf{D}} in its label to $C$), the merge is complete and the new component $CC'$ is ready to begin a new comparison. When the merge is complete, $u$ informs the root by sending a merge--complete message to its parent. The merge--complete message is then forwarded by each particle that receives it to its parent until the root is reached. After a merge, $CC'$ marks all endpoints of DBEs (i.e., particles with $D$ in their label) as ``finished'' and waits until the remaining endpoint for each of the DBEs is also marked as finished. When all neighbouring components have finished the current comparison, $CC'$ begins executing ``Tree Traversal'' again, until no more edges marked as not--compared remain. While $CC'$ only has \textit{compared} edges it waits until an edge becomes not--compared as described in ``Tree Comparisons''. This is also shown in Procedure \ref{proc:restart-competition}. A component participating in a merge as a losing (resp. winning) component might be simultaneously participating in a merge on a different DBE as a winning (resp. losing or winning) component. A component participates in at most one merge as \textit{losing}, so the exchanged messages in the above procedure are unique. Additionally, a component simultaneously participating in two merges as \textit{winning} does so only through the two particles incident to the chosen DBEs and the resulting changes in the tree are local changes in the respective positions of the labels of those particles. It is possible that a component consisting of only one particle that is only connected to components consisting of one particle each, is a \textit{winning} component in both merges. In this case, the single particle of the component implements both winning merges as described above, by recording the merges in different positions of its own label. If the DBE neighbour of a winning component marks itself as \textit{finished} instead of merging with the winning component, it is inferred that the losing component merged on a different DBE. The particle of the winning component incident to the DBE on which a comparison is made also marks itself as \textit{finished}, marks the DBE as \textit{not--compared} and informs the root that the merge is complete through its parent. Finally, a component participating in two simultaneous merges both as losing and as winning, participates in both merges independently since even if the particle changing its label for the winning merge participates in the losing merge, that must be through different ports and as a result, different parts of its label are changed for each of the merges. The pseudocode corresponding to the merging of components is shown in Procedure \ref{proc:merge-components}. 

\begin{procedure}[ht]
    \caption{MergeComponents()}
    \label{proc:merge-components}

    States: \{{\small \textbf{WinningMerge}, \textbf{LosingMerge}, \textbf{InMerge}}\} \\
    
    \underline{In state \texttt{WinningMerge}:} \\
    \quad wait for all competitions in the component to finish \\
    \quad \textsc{MergeInitiator} ( 
    \textit{DBE neighbour marks this particle as a parent}: \label{line:marked-as-parent} \\ 
        \hspace{5.5cm} mark the DBE neighbour as a child, \label{line:mark-child} \\
        \hspace{5.5cm} send merge--complete message to parent, \\ 
        \hspace{5.5cm} \texttt{MergeComplete} ) \\

    \underline{In state \texttt{LosingMerge}:} \\
    \quad wait for all competitions in the component to finish\\
    \quad send in--merge message to parent \\
    \quad \textsc{MergeInitiator} ( 
    \textit{receive \emph{abort--merge} message from root}: \texttt{NotActive}; \\
    \hspace{3.73cm} \textit{receive \emph{in--merge} message}: \texttt{MergeComplete} ) \\

    \underline{In state \texttt{InMerge}: } \\
    \quad \textsc{InLosingMerge} ( \\
    \hspace{.7cm} wait for all competitions in the component to finish\\
    \hspace{.7cm} \textit{is the root \emph{AND} received one \emph{in--merge} message}: \\
        \hspace{1.7cm} mark the child that forwarded the in--merge message as  the parent, \label{line:root-gets-parent} \\
        \hspace{1.7cm} send the in--merge message to the parent, \\
        \hspace{1.7cm} \texttt{MergeComplete}; \\
    \hspace{.7cm} \textit{is the root \emph{AND} received two \emph{in--merge} messages}: \label{line:lose-two-condition} \\ 
        \hspace{1.7cm} mark the child that forwarded the \emph{in--merge} message for the outgoing \label{line:root-gets-parent-too}\\
        \hspace{2cm} competition as the parent, \\
        \hspace{1.7cm} send the \emph{in--merge} message to the parent, \\
        \hspace{1.7cm} send \emph{abort--merge} message to particle in incoming DBE, \label{line:abort-merge} \\
        \hspace{1.7cm} \texttt{MergeComplete}; \\
    \hspace{.7cm}\textit{ is not the root}: \\
         \hspace{1.7cm} mark the parent as a child, \\
         \hspace{1.7cm} mark the child that forwarded the in--merge message as the parent, \\
         \hspace{1.7cm} send the in--merge message to the parent, \\
         \hspace{1.7cm} \texttt{MergeComplete} ) \\

\end{procedure}

\begin{procedure}[h]
    \caption{RestartCompetition()}
    \label{proc:restart-competition}

    States: \{\textbf{\small MergeComplete}\} \\

    \underline{In state \texttt{MergeComplete}:} \\
    \quad \textsc{RestartingCompetition} ( \\ 
    \qquad \textit{no \emph{not--compared} edges exist in the component}: \texttt{MergeComplete}; \\
    \qquad \textit{\emph{not--compared} edges exist in the component}: \texttt{OutgoingEdge} ) \\
    
\end{procedure}

\subsubsection*{Correctness}

We show that the encoding of a competing component is not altered during the execution of LFC (Lemma \ref{lem:moving-labels-around-tree}). This ensures that each of the particles incident to an active DBE of a component receives the same encoding during a comparison. A comparison between two components always terminates (Lemma \ref{lem:tree-comparisons-correctness}) and the result of a comparison is always \textit{win}, \textit{lose} or \textit{draw} (Lemma \ref{lem:merging-two-components}). In Lemma \ref{lem:tree-comparisons-correctness} we further show that a component wins on a DBE if and only if the second component competing on the same edge loses. Similarly, if a component draws on an edge the other component competing on the same edge also draws. Next, we show that progress is always ensured during the execution of the algorithm since either the number of competing components decreases or the number of DBEs on which comparisons can be performed decreases. Eventually no DBEs on which a competition can be performed exist and when this happens only one component with one leader remains (Theorem \ref{th:leader-elected}).  

\begin{lemma} \label{lem:moving-labels-around-tree}
    Let $T$ be a tree encoding a component. The information held by any particle $p \in T$, can be moved around the tree until any other particle, $p' \in T$, is reached and the order of information held by neighbouring particles in the tree is preserved, using the procedure described in ``Tree Traversal''.  
\end{lemma}
\begin{proof}
    In Tree Traversal, each node $v$ of the tree is split into \emph{number of children}$ + 1$ agents, where agent$_1$ denotes the first time $v$ is visited in cyclic--DFS and agent$_i$ for $i > 1$ denotes that $i - 1$ children of $v$ have been visited in cyclic--DFS. Each of the agents corresponds to a node in the equivalent ring, as shown in Figure \ref{fig:tree-traversal}. Due to the resulting ring, each agent has exactly one predecessor and one successor so it receives and forwards information, to the same agents. A label received at some step is not forwarded before a label received at an earlier step and as a result the order of labels in the ring is preserved.  
\end{proof}

\begin{lemma} \label{lem:tree-comparisons-correctness}
    Let $T_1$ and $T_2$ be the tree encodings of two components connected by at least one DBE \emph{$p_{db}p_{db}'$} such that $p_{db} \in T_1$ and $p_{db}' \in T_2$. Using the procedure described in ``Tree Comparisons'' $p_{db}$ and $p'_{db}$ learn whether the encodings of $T_1$ and $T_2$ are the same after at most $\mathcal{O}(\min\{|T_1|, |T_2|\})$ comparisons.
\end{lemma}
\begin{proof}
    Let $r_1$ be the root of $T_1$ and $r_2$ be the root of $T_2$. As we have shown in Lemma \ref{lem:moving-labels-around-tree} information in the tree can be moved between any two particles in the tree. In the first step of ``Tree Traversal'', when the label marked as root--label reaches $p_{db}$ (resp. $p'_{db}$) the label forwarding stops. Since $p_{db}$ and $p'_{db}$ wait until both particles are marked as \textit{ready} to initiate ``Tree Comparisons'', when ``Tree Comparisons'' starts $p_{db}$ and $p'_{db}$ hold the root--label of $r_1$ and $r_2$ respectively. Since we know from Lemma \ref{lem:moving-labels-around-tree} that the order of labels in the tree is preserved and $p_{db}, p'_{db}$ always obtain, forward and compare labels in the same order, the compared labels at each step, belong to equivalent nodes in $T_1$ and $T_2$ if $T_1, T_2$ have the same encoding. Since all labels are compared in the same order and the comparison in ``Tree Comparisons'' stops only if a difference in a label or on the tree size is detected, $p_{db}$ and $p'_{db}$ have the same information which means that they must also make the same decision about the result of the comparison. 
\end{proof}

\begin{lemma} \label{lem:merging-two-components}
    Let two components $C_1$ and $C_2$ be connected by at least one DBE $uv$. Merging $C_1$ and $C_2$ on $uv$ using ``Merging Components'' creates one component, $C_1C_2$. The encoding of $C_1C_2$ is still a tree and exactly one potential leader, the root of $C_1C_2$, remains on the merged tree. 
\end{lemma}
\begin{proof}
    Call the tree encoding of the components $C_1$ and $C_2$, $T_1 = (V_1,E_1)$ and $T_2 = (V_2,E_2)$ respectively. Since $C_1$ and $C_2$ are two distinct components, by definition $V_1 \cap V_2 = \emptyset$. Also, since we have taken two components $C_1$ and $C_2$ that are connected by a DBE, without loss of generality, take $uv$ such that $u \in V_1$ and $v \in V_2$ and say $C_1$ is the winning component. When a path of consecutive edges is reversed, $T_2$ remains a tree. Furthermore, since $v$ marks itself as a child of $u$, $T_2$ becomes connected to $T_1$ by a single directed edge. When $u$ marks $v$ as its child, $T_1$ also becomes connected to $T_2$ by the same directed edge. Since only one directed edge was added to two previously disconnected trees, a cycle cannot have been formed and the resulting graph is also a tree. Finally, since the root of $T_2$ becomes a follower and no other particle of $T_2$ becomes a leader in ``Merging Components'', only the root of $T_1$ remains at state \textit{leader} after the merge. Therefore, the combined graph $T_1T_2$ is a tree and contains only one leader after the merge.   
\end{proof}

Notice that since a component participates in at most two simultaneous comparisons the component might win both comparisons, lose one comparison and win one comparison, win or lose one comparison and draw on the other, draw in both comparisons or lose both comparisons. In the first three cases, the changes corresponding to different merges of the component are independent of each other. A component drawing in both competitions only marks edges as compared without changing the encoding of the component. When losing both comparisons, the root of the component only participates in one merge and the second merge is aborted. After a comparison between two components, either the components merge or they do not merge but the DBE on which the competition is made is marked as \textit{compared}. 

We define a \textit{rightmost DBE} $ab$ to be a DBE such that there does not exist a particle $b'$ that is an endpoint of another DBE with $x_{b'} > x_{b}$. We further define a \textit{topmost} DBE $ab$ to be a rightmost DBE such that there does not exist a $b'$ that is an endpoint of another rightmost DBE with $y_{b'} > y_{b}$. We call an edge connecting particles in the same component (resp. different components) \textit{internal} (resp. \textit{external}). A DBE where a comparison is being made is \textit{active}.

\begin{lemma} \label{lem:internal-db}
    Let $e$ be an internal DBE of some component $C$ and let $e'$ be a different DBE of $C$. For every comparison it is not possible that $e$ and $e'$ are both active. 
\end{lemma}
\begin{proof}
    At most one incoming and at most one outgoing DBE become active in each component during a comparison. So, $e$ must be chosen as the outgoing DBE, otherwise, it does not participate in the comparison. However, if $e$ is chosen as an outgoing DBE for $C$ it must either be chosen as an incoming edge too or it does not participate in the comparison. Therefore, either $e$ is only chosen as an outgoing edge and as a result it is not an active edge, or $e$ is chosen both as an incoming and an outgoing DBE but no other edges are chosen for the same comparison. In either case, it is not possible for both $e$ and any other edge $e'$ to be simultaneously active in $C$.
\end{proof}

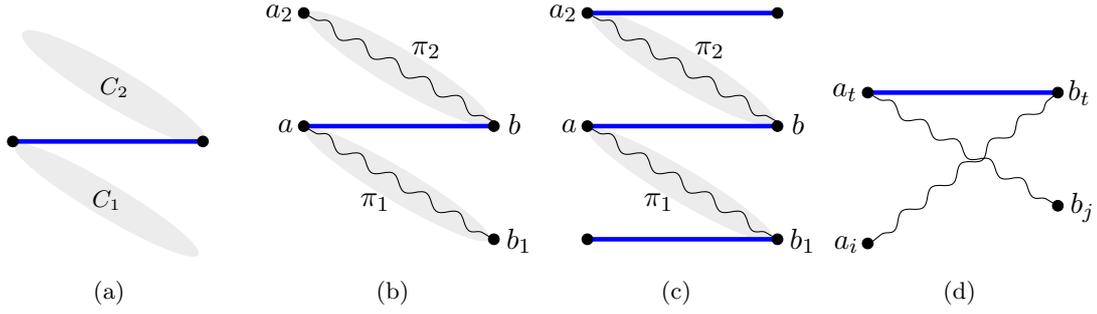
\begin{figure}[t]
    \centering
    \begin{subfigure}[b]{0.24\textwidth}
        \centering
        \begin{tikzpicture} [decoration={snake,segment length=5mm,amplitude=.7mm},
   line around/.style={decoration={pre length=#1,post length=#1}}]

            \filldraw[rotate around={59:(-.2,-.15)},gray!15] (0,-1.7) node[]{\textcolor{black}{\small $C_1$}} ellipse (6pt and 1.4cm) ;

            \filldraw[rotate around={59:(-.15,-.2)},gray!15] (1.4,-1) node[]{\textcolor{black}{\small $C_2$}} ellipse (7pt and 1.4cm);

            \draw[blue,line width=.6mm] (0:0) --++ (0:2.5) ;
            
            \filldraw[black] (0:0) circle(2pt) ++ (0:2.5) circle(2pt) ;

        \end{tikzpicture}
        \caption{}
        \label{fig:main-lemma-DBE}
    \end{subfigure}
    \begin{subfigure}[b]{0.24\textwidth}
        \centering
        \begin{tikzpicture}[decoration={snake,segment length=5mm,amplitude=.7mm},
   line around/.style={decoration={pre length=#1,post length=#1}}]

            \filldraw[rotate around={59:(-.2,-.15)},gray!15] (0,-1.7) ellipse (6pt and 1.4cm) ;

            \filldraw[rotate around={59:(-.15,-.2)},gray!15] (1.4,-.9) ellipse (7pt and 1.4cm);

            \draw[blue,line width=.6mm] (0:0) --++ (0:2.5) ;
            
            \filldraw[black] (0:0) circle(2pt) ++ (0:2.5) circle(2pt) ;

            \filldraw[black] (90:1.5) circle(2pt) ++ (0:2.5) ++ (-90:3) circle(2pt) ;

            \draw[decorate,line around=5pt] (0:0) node[left] {$a$} --++ (-31:2.95) node[right] {$b_1$} node [midway,below left,draw=none] {$\pi_1$} ;

            \draw[decorate,line around=5pt] (0:0) ++ (90:1.5) node[left] {$a_2$} --++ (-30:2.95) node[right] {$b$} node [midway,above right,draw=none] {$\pi_2$} ;
        \end{tikzpicture}
        \caption{}
        \label{fig:main-lemma-paths}
    \end{subfigure}
    \begin{subfigure}[b]{0.24\textwidth}
        \centering
        \begin{tikzpicture}[decoration={snake,segment length=5mm,amplitude=.7mm},
   line around/.style={decoration={pre length=#1,post length=#1}}]

            \filldraw[rotate around={59:(-.2,-.15)},gray!15] (0,-1.7) ellipse (6pt and 1.4cm) ;

            \filldraw[rotate around={59:(-.15,-.2)},gray!15] (1.4,-.9) ellipse (7pt and 1.4cm);

            \draw[blue,line width=.6mm] (0:0) --++ (0:2.5) ;

            \draw[blue,line width=.6mm] (90:1.5) --++ (0:2.5) ;

            \draw[blue,line width=.6mm] (-90:1.5) --++ (0:2.5) ;
            
            \filldraw[black] (0:0) circle(2pt) ++ (0:2.5) circle(2pt) ;

            \filldraw[black] (90:1.5) circle(2pt) ++ (0:2.5) circle(2pt) ++ (-90:3) circle(2pt) ++ (180:2.5) circle(2pt) ;

            \draw[decorate,line around=5pt] (0:0) node[left] {$a$} --++ (-31:2.95) node[right] {$b_1$} node [midway,below left,draw=none] {$\pi_1$} ;

            \draw[decorate,line around=5pt] (0:0) ++ (90:1.5) node[left] {$a_2$} --++ (-30:2.95) node[right] {$b$} node [midway,above right,draw=none] {$\pi_2$} ;
            
        \end{tikzpicture}
        \caption{}
        \label{fig:main-lemma-three-DBEs}
    \end{subfigure}
    \begin{subfigure}[b]{0.24\textwidth}
        \centering
        \begin{tikzpicture}[decoration={snake,segment length=5mm,amplitude=.7mm},
   line around/.style={decoration={pre length=#1,post length=#1}}]

            \draw[blue,line width=.6mm] (0:0) --++ (0:2.5) ;
   
            \filldraw[black] (0:0) circle(2pt) ++ (0:2.5) circle(2pt) ++ (-90:1.5) circle(2pt) ++ (180:2.5) ++ (-90:.5) circle(2pt)  ;

            \draw[decorate,line around=5pt] (0:0) node[left] {$a_t$} --++ (-31:2.95) node[right] {$b_j$} node [midway,below left,draw=none] {} ;

            \draw[decorate,line around=5pt] (0:2.5) node[right] {$b_t$} --++ (219:3.2) node[left] {$a_i$} node [midway,below left,draw=none] {} ;
        \end{tikzpicture}
        \caption{}
        \label{fig:main-lemma-crossing-paths}
    \end{subfigure}
    \caption{Visualisation of the setting described in the proof of Lemma \ref{lem:final-config-DBEs-internal}}
    \label{fig:main-lemma}
\end{figure}
\begin{lemma} \label{lem:final-config-DBEs-internal}
    All DBEs in the system are internal in the final configuration, that is, the final configuration consists of only one component.
\end{lemma}
\begin{proof}
    Let us suppose that there exists at least one DBE in the final configuration which remains external after the execution of the algorithm. Out of the edges that remain external in the final configuration, take $e$ to be the rightmost topmost one. This means that $e$ connects two components of the final configuration which did not merge. Call the configuration on which $e$ was marked as \textit{active} for the last time $\mathcal{F}$ and call the components connected by $e$ in $\mathcal{F}$, $C_1$ and $C_2$. Notice that all edges that are external and incident to $C_1$ or $C_2$ in $\mathcal{F}$ are also external in the final configuration, otherwise, $\mathcal{F}$ is not the last configuration in which $e$ is \textit{active}. Out of the active external DBEs connecting $C_1$ and $C_2$ in $\mathcal{F}$ take $e = ab$ to be a rightmost DBE such that $a \in C_1$ and $b \in C_2$. The edge $ab$ and the two components are those in Figure \ref{fig:main-lemma-DBE}. In our algorithm components with different encodings merge, so the tree encodings of $C_1$ and $C_2$ (call them $T_1$ and $T_2$) must be the same. Since $a$ and $b$ are both marked as endpoints of an active DBE in $\mathcal{F}$, when $a$ (resp. $b$) is compared to its equivalent particle in $T_2$ (resp. $T_1$), the equivalent particle $a_2$ (resp. $b_1$), must also be the endpoint of an active DBE. So there must be a path $\pi_1$ in $T_1$ (resp. $\pi_2$ in $T_2$) connecting $a$ (resp. $b$) to $b_1$ (resp. $a_2$) (Figure \ref{fig:main-lemma-paths}). Due to Lemma \ref{lem:internal-db} we know that the DBEs incident to $a_2$ and $b_1$ (Figure \ref{fig:main-lemma-three-DBEs}) cannot be internal since they would not be active simultaneously with $ab$. We further know that $\pi_1$ must be equivalent to $\pi_2$ since we have assumed that the encodings of $C_1$ and $C_2$ are the same. Since the paths are equivalent, the horizontal shift between $a, b_1$ and $a_2, b$ must also be equal. We then have $x_a - x_{b_1} = x_{a_2} - x_b$ or equivalently $x_a + x_b = x_{a_2} + x_{b_1}$. Since $ab$ is a rightmost DBE, it must be that $x_{a_2} \leq x_a$ and $x_{b_1} \leq x_b$. Suppose $x_{a_2} < x_a$. Then to maintain the above equality, it must be that $x_{b_1} > x_b$ which is a contradiction, so we must have $x_{a_2} = x_a$. Similarly, we also get $x_{b_1} = x_b$. We have hence proved that $T_1$ (resp. $T_2$) is incident to two active, external, rightmost DBEs that are also connected to each other by paths in $T_1$ and $T_2$ in $\mathcal{F}$.   

    Out of all rightmost DBEs in $\mathcal{F}$, we now take $a_tb_t$ to be the topmost DBE and following the argumentation of the previous paragraph, there must exist a path from $b_t$ to a lower particle $a_i$ (i.e., $y_{a_i} < y_{b_t}$) and a path from $a_t$ to a lower particle $b_j$ (i.e., $y_{b_j} < y_{a_t}$). Take the paths starting at the $a$ particles: $a_t$ and $a_i$. Since all particles by definition agree on the left and right directions, every time the path starting at $a_t$ moves right (resp. left) the path starting at $a_i$ also moves right (resp. left) or equivalently the paths starting at rightmost $a$ particles start from the same $x$ coordinate and always maintain the same $x$ coordinates. However, the path starting at $a_t$ (resp. $a_i$) must eventually move down (resp. up) to reach $b_j$ (resp. $b_t$) that is not (resp. that is) at the topmost DBE. Since the two equal paths start from the same $x$ coordinate and $y_{b_t} > y_{b_j}$, the paths must cross (Figure \ref{fig:main-lemma-crossing-paths}). However, if the paths cross there is a path connecting $a_t$ to $a_i$ so $a_t$ and $a_i$ must be in the same tree encoding of a component. We have defined $a$ particles to be endpoints to active outgoing DBEs and we know that there is only one active outgoing DBE in each component, so this setting is not possible and $C_1$, $C_2$ must have different tree encodings. Hence a merge occurs either between $C_1$ and $C_2$ (which we have assumed not to be the case) or between one of $C_1,C_2$ and a neighbouring component. In the latter case, $e$ must become \textit{active} in some future round since the components have changed, which contradicts the assumption that $\mathcal{F}$ is the last time $e$ is active. In either case, there cannot exist a DBE which is external in the final configuration so the system eventually consists of only one component with one leader.
\end{proof}

\begin{theorem} \label{th:leader-elected}
    Starting from any number of grey components encoded as trees, after executing the above algorithm only one tree with one leader, remains.
\end{theorem}
\begin{proof}
    In each round performed by a component, either an external DBE is chosen and a merge occurs in the system or an internal edge is chosen and it becomes \textit{compared}. Consequently the system stabilizes in a final configuration where all DBEs are compared. From Lemma \ref{lem:final-config-DBEs-internal} we know that this is possible only if there is only one component. 
    Therefore, eventually the system consists of only one component with one leader.
\end{proof}

\section{Complexity Discussion}
We give an upper bound for the complexity of our algorithm. We measure the complexity in \textit{activation units}. In each \textit{activation unit}, every particle is activated at least once. Notice that what we call an \textit{activation unit} is usually called a \textit{round} in the literature, but we want to make the distinction between ``activation units'' in this section and ``rounds'' in Section \ref{sec:component-competition} clear. For the Grey Leader Election phase, we have the following. The algorithm from \cite{bazzi2019stationary} works in $O(n^2)$ \textit{activation units}, where $n$ is the number of particles in the system. The identifier calculation phase needs constant time since each of the at most six heads makes one local calculation for each head on the boundary. Then, to decide on a unique head for a grey component, a message is forwarded through at most every particle on the outer boundary of its component. It is possible that all particles of a component are its outer boundary and that all particles in the system form a single grey component. Therefore, the final phase of electing a unique head on a grey component needs $O(n)$ \textit{activation units} and Grey LE requires $O(n^2)$ \textit{activation units} in total. 

In the Tree Construction part of the algorithm, each particle in the component needs to wait until a path of particles between it and the root enter the tree and each component contains at most $n$ particles. Therefore, Tree Construction needs $O(n)$ \textit{activation units}. 

In Component Competition, an LFC procedure moves labels around a virtual ring obtained by the cyclic--DFS traversal of the tree. By construction of the virtual ring, each node of the tree corresponds to at most seven agents. Furthermore, we have already provided an intuition about why a label at distance $i$ from the agent initiating the LFC needs $2i$ \textit{activation units} to reach the initiating agent and we know that the maximum length of the ring is $6n + 1$ for the case where all particles in the system are in a common component. Therefore, LFC also needs a linear number of \textit{activation units} on the number of particles in the system. Furthermore, at most two LFCs are performed by at most two particles in each component per round. To merge two components, one path of particles in one of the components reverses its parent-child edges and since the maximum length of a single path in the tree encoding of a component is the height of the tree, merging is done in $O(n)$ \textit{activation units}. We have defined a \textit{round for a component} (i.e., the definition of round in Section \ref{sec:component-competition}) to be one iteration of Tree Traversal, Tree Comparisons and Merging Components. We know that up to $n$ components can exist in the system (i.e., each particle is a component). Furthermore, we have proved that at least one merge is performed when an external edge is selected in a round. However, between two consecutive selections of external DBEs, up to $n$ selections of internal DBEs can occur. This is because after a merge occurs all \textit{compared} edges of the merged components become \textit{not--compared} again. So at most $n-1$ rounds of external DBE activations (one for each component to merge with a neighbouring component) are needed and up to $n$ rounds are needed between two consecutive activations of external DBEs. Therefore, we need $n(n-1)$ rounds each of which needs $O(n)$ \textit{activation units}, or $O(n^3)$ \textit{activation units} in total for the Component Competition part of the algorithm.

Summing up the \textit{activation units} needed for Grey Leader Election, Tree Construction and Component Competition which are performed successively, our algorithm works in $O(n^3)$ \textit{activation units}.

We now show that this upper bound is tight for the analysis of our algorithm. Consider the system $\mathcal{S}$ of Figure \ref{fig:complexity}. $\mathcal{S}$ consists of $n$ particles: two ``big'' components $C_1$ and $C_2$ consisting of $\frac{n}{4}$ particles each and $\frac{n}{2}$ ``small'' components consisting of one particle each. Assume that all components are connected by DBEs as shown in Figure \ref{fig:complexity}. Furthermore, assume that each of $C_1$ and $C_2$ contains $O(n)$ internal DBEs and that the internal DBEs are placed in such positions that they are always chosen before the external DBE. An example of how to construct a particle system with $O(n)$ internal DBEs is given in Figure \ref{fig:linear-internal-DBEs}.

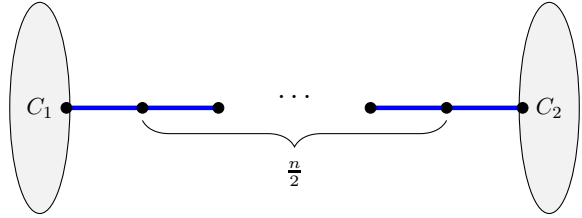
\begin{figure}[ht]
    \centering
    \scalebox{1}{
    \begin{tikzpicture}

        \node[ellipse, minimum height=3.5cm,minimum width=.8cm,draw=black,fill=gray!10,scale=0.8] (A) at (.15,0.865) {$C_1$};

        \node[ellipse, minimum height=3.5cm,minimum width=.8cm,draw=black,fill=gray!10,scale=0.8] (A) at (6.85,0.865) {$C_2$};

        \draw[line width=.6mm, blue] (60:1) --++ (0:1) --++ (0:1) ++ (0:2) --++ (0:1) --++ (0:1) ;

        \filldraw[black] (60:1) circle(2pt) ++ (0:1) circle(2pt) ++ (0:1) circle(2pt) ++ (0:1) node[above] {$\ldots$} ++ (0:1) circle(2pt) ++ (0:1) circle(2pt) ++ (0:1) circle(2pt) ;

        \draw [decorate,decoration={brace,amplitude=10pt},xshift=0pt,yshift=0pt] (5.5,0.7) -- (1.5,0.7) node [black,midway,yshift=-20pt] {\footnotesize $\frac{n}{2}$};

    \end{tikzpicture}
    }
    \caption{A particle system consisting of $\frac{n}{2} + 2$ grey components and $n$ particles.}
    \label{fig:complexity}
\end{figure}
\begin{figure}[ht]
    \centering
    \scalebox{1}{
    \begin{tikzpicture}

        \triangularGrid{2}{6}{1}{2*0.4325}

        \draw[line width=.3mm, blue, dashed, very thick] (60:1) --++ (0:1);

        \draw[line width=.3mm, gray!60, dashed, very thick] (0:0) ++ (60:1) --++ (120:1) --++ (60:1) --++ (60:1) --++ (-60:1) --++ (-60:1) --++ (-120:1) ++ (-60:1);

        \draw[line width=.3mm, cyan, dashed, very thick] (60:3) --++ (180:1) ;

        \filldraw[black] (0:0) ++ (60:1) circle(1pt) ++ (120:1) circle (1pt) ++ (60:1) circle (1pt) ++ (60:1) circle (1pt) ++ (-60:1) circle (1pt) ++ (-60:1) circle(1pt) ++ (-120:1) circle (1pt) ;

        \draw[line width=.3mm, blue,dashed, very thick] (60:1) ++ (0:4) --++ (0:1);

        \draw[line width=.4mm, gray!60,dashed, very thick] (0:4) ++ (60:1) --++ (120:1) --++ (60:1) --++ (60:1) --++ (-60:1) --++ (-60:1) --++ (-120:1) ++ (-60:1);

        \draw[line width=.3mm, cyan, dashed, very thick] (60:3) ++ (0:3) --++ (0:1) ;

        \filldraw[black] (0:4) ++ (60:1) circle(1pt) ++ (120:1) circle (1pt) ++ (60:1) circle (1pt) ++ (60:1) circle (1pt) ++ (-60:1) circle (1pt) ++ (-60:1) circle(1pt) ++ (-120:1) circle (1pt) ++ (-60:1) ;

        \draw[line width=.6mm, blue] (60:1) ++ (0:2) --++ (0:1);

        \draw[line width=.8mm, gray!70] (0:2) ++ (60:1) --++ (120:1) --++ (60:1) --++ (60:1) --++ (-60:1) --++ (-60:1) --++ (-120:1) ++ (-60:1);

        \draw[line width=.6mm, cyan] (60:3) ++ (0:1) --++ (0:1);

        \filldraw[black] (0:2) ++ (60:1) circle(2pt) ++ (120:1) circle (2pt) ++ (60:1) circle (2pt) ++ (60:1) circle (2pt) ++ (-60:1) circle (2pt) ++ (-60:1) circle(2pt) ++ (-120:1) circle (2pt) ++ (-60:1) ;

    \end{tikzpicture}
    }
    \caption{In the middle we give a ``base'' particle configuration that can be repeated (as indicated by the dashed components) any number of times to give an arbitrarily large grey component with $O(n)$ internal DBEs, where $n$ is the number of particles in the grey component.}
    \label{fig:linear-internal-DBEs}
\end{figure}
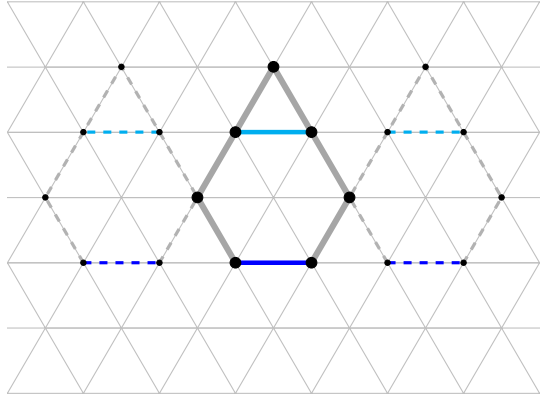

From Lemma \ref{lem:tree-comparisons-correctness} we know that a comparison between two trees $T_1$ and $T_2$ needs $O(\min\{|T_1|,|T_2|\})$ activation units. Furthermore, in the case of a comparison on an internal DBE, $T_1$ and $T_2$ correspond to the same tree and as a result for $i \in \{1,2\}$, $\min\{|T_i|,|T_i|\} = |C_i| = \frac{n}{4}$ comparisons are made in $\mathcal{S}$ for each of the $O(n)$ internal DBEs before an external DBE is chosen. By construction of $\mathcal{S}$, we know that there exist $\frac{n}{2}$ external DBEs and since $|C_1| = |C_2|$, $C_1$ (resp $C_2$) merges on $\frac{n}{4}$ of those external DBEs, after performing one comparison on each, before competing with $C_2$ (resp. $C_1$). So $2*O(n)*\frac{n}{4}*\frac{n}{4}$ comparisons are necessary in $\mathcal{S}$ before $C_1$ (resp. $C_2$) competes with $C_2$ (resp. $C_1$), or equivalently $O(n^3)$ activation units.

\section{Conclusion and Open Problems}
We presented an algorithm that deterministically solves the LE problem for stationary particles that agree on one direction. We showed that the difficulties in this case are complementary to those of particles with common chirality. In the former case a unique leader can be elected by an implicit termination algorithm, whereas, in the latter case an explicit termination algorithm is possible but up to six leaders are elected. Notice that stationary particles agreeing on a common direction cannot agree on chirality even after LE, as chirality cannot be communicated along DBEs. A natural next step is determining minimal sets of properties of an initial configuration that would allow a stationary, deterministic and explicitly terminating LE algorithm that elects a unique leader. For example, in the case of particles agreeing on a common direction, if the system does not contain DBEs, Grey LE solves the problem. If only one DBE exists in the system, due to the common direction, one of the endpoints of the DBE can be elected. Can a terminating LE algorithm be found for any constant number of DBEs? An alternative, is to study the problem in systems where the number of DBEs is not constant but each grey component has at most one incoming and at most one outgoing DBE. Another assumption is to restrict to configurations where each particle is in at most one or at most two boundaries, instead of the three boundaries that we consider here. Finally, it would be interesting to consider more problems, such as Shape Formation, under the assumption of common direction.

\bibliographystyle{splncs04}
\bibliography{short_biblio}

\end{document}